\pdfoutput=1
\RequirePackage{silence}
\documentclass[a4paper,USenglish,cleveref,autoref,thm-restate]{lipics-v2021}

\newenvironment{algorithm}[3]{%
\bigskip
\noindent\textbf{Algorithm #1}({\itshape#2\/}) {\itshape#3\/}
\begin{description}}{%
\end{description}\medskip}

\usepackage{mathtools}
\DeclarePairedDelimiter\paren{\lparen}{\rparen}
\DeclarePairedDelimiter\abs{\lvert}{\rvert}
\DeclarePairedDelimiter\set{\{}{\}}
\DeclarePairedDelimiterX\setc[2]{\{}{\}}{\,#1 \;\colon\; #2\,}
\DeclarePairedDelimiterX\parenc[2]{\lparen}{\rparen}{\,#1 \;\delimsize\vert\; #2\,}

\usepackage{tikz}
\usetikzlibrary{shapes.geometric,calc, decorations.pathreplacing}
\tikzstyle{sq}=[fill,thick,inner sep=0pt,shape=rectangle,minimum size=1.5mm]
\tikzstyle{ci}=[fill,thick,inner sep=0pt,shape=circle,minimum size=1.5mm]
\tikzstyle{di}=[fill,thick,inner sep=0pt,shape=diamond,minimum size=2.12132mm]

\newcommand{\cc}[1]{\ensuremath{\mathrm{#1}}}
\newcommand{\pp}[1]{\textup{#1}}
\newcommand{\op}[1]{\ensuremath{\operatorname{#1}}}

\newcommand{\sharpP}{\cc{\#P}}
\newcommand{\FPT}{\cc{FPT}}
\newcommand{\W}{\cc{W[1]}}
\newcommand{\pW}{\oplus\W}
\newcommand{\sharpW}{\cc{\#W[1]}}

\newcommand{\ETH}{\ensuremath{{\mathrm{ETH}}}}

\newcommand{\poly}{\op{poly}}

\newcommand{\N}{\mathbf{N}}
\newcommand{\Z}{\mathbf{Z}}

\newcommand{\dotcup}{\mathbin{\dot\cup}}
\newcommand{\zo}{\set{0,1}}

\DeclareDocumentCommand{\restrict}{O{}}{\mathord{\restriction}_{#1}}

\newcommand{\classH}{\ensuremath{{\mathcal{H}}}}

\newcommand{\Aut}{\ensuremath{{\mathrm{Aut}}}}
\newcommand{\Sub}{\ensuremath{{\mathrm{Sub}}}}
\newcommand{\vcSub}{\ensuremath{{\mathrm{VertexColorfulSub}}}}

\newcommand{\Emb}{\ensuremath{{\mathrm{Emb}}}}
\newcommand{\id}{\ensuremath{{\mathrm{id}}}}
\newcommand{\Mod}{\ensuremath{{\mathrm{Mod}}}}

\newcommand{\ModSub}[2]{\ensuremath{\#\Sub(#1)\bmod #2}}

\newcommand{\haf}{\operatorname{haf}}

\newcommand{\vc}{\operatorname{vc}}

\title{Modular counting of subgraphs: \\Matchings, matching-splittable graphs, and paths}
\titlerunning{Modular counting of subgraphs}

\author{Radu Curticapean}{IT University of Copenhagen, Denmark \and Basic Algorithms Research Copenhagen (BARC), Denmark \and \url{https://www-cc.cs.uni-saarland.de/curticapean/}}{racu@itu.dk}{https://orcid.org/0000-0001-7201-9905}{}

\author{Holger Dell}{Goethe University Frankfurt, Germany \and IT University of Copenhagen, Denmark \and Basic Algorithms Research Copenhagen (BARC), Denmark \and \url{https://holgerdell.com/}}{hold@itu.dk}{https://orcid.org/0000-0001-8955-0786}{}

\author{Thore Husfeldt}{IT University of Copenhagen, Denmark \and Basic Algorithms Research Copenhagen (BARC), Denmark \and Lund University, Sweden \and \url{https://thorehusfeldt.com/}}{hold@itu.dk}{https://orcid.org/0000-0001-9078-4512}{}

\authorrunning{R. Curticapean, H. Dell, and T. Husfeldt}

\Copyright{Radu Curticapean, Holger Dell, and Thore Husfeldt}

\ccsdesc[500]{Theory of computation~Fixed parameter tractability}
\ccsdesc[300]{Theory of computation~Problems, reductions and completeness}

\keywords{Counting, matchings, paths, parameterized complexity} 

\nolinenumbers 

\hideLIPIcs  

\EventEditors{John Q. Open and Joan R. Access}
\EventNoEds{2}
\EventLongTitle{42nd Conference on Very Important Topics (CVIT 2016)}
\EventShortTitle{CVIT 2016}
\EventAcronym{CVIT}
\EventYear{2016}
\EventDate{December 24--27, 2016}
\EventLocation{Little Whinging, United Kingdom}
\EventLogo{}
\SeriesVolume{42}
\ArticleNo{23}

\begin{document}

\maketitle

\begin{abstract}
  We systematically investigate the complexity of counting subgraph patterns \emph{modulo fixed integers}.
  For example, it is known that the \emph{parity} of the number of $k$-matchings can be determined in polynomial time by a simple reduction to the determinant.
  We generalize this to an $n^{f(t,s)}$-time algorithm to compute modulo~$2^t$ the number of subgraph occurrences of patterns that are $s$ vertices away from being matchings.
  This shows that the known polynomial-time cases of subgraph \emph{detection} (Jansen and Marx, SODA 2015\nocite{JansenMarx}) carry over into the setting of \emph{counting modulo~$2^t$}.
  Complementing our algorithm, we also give a simple and self-contained proof that counting $k$-matchings modulo odd integers $q$ is $\Mod_q \W$-complete and prove that counting $k$-paths modulo $2$ is $\pW$-complete, answering an open question by Björklund, Dell, and Husfeldt~(ICALP~2015)\nocite{DBLP:conf/icalp/BjorklundDH15}.
\end{abstract}


\section{Introduction}

The last two decades have seen the development of several complexity dichotomies for pattern counting problems in graphs, including full classifications for counting \emph{subgraphs}, \emph{induced subgraphs}, and \emph{homomorphisms} from fixed computable pattern classes $\mathcal H$. 
The input to such problems is a \emph{pattern} graph $H \in \mathcal H$ and an unrestricted \emph{host} graph $G$; 
the task is to count the relevant occurrences of $H$ in $G$. 
Depending on $\mathcal H$, these problems are known to be either polynomial-time solvable or $\sharpW$-hard when parameterized by $|V(H)|$.
The latter rules out polynomial-time algorithms under the complexity assumption $\FPT \neq \sharpW$.

In this paper, we focus on counting \emph{subgraphs} from any fixed graph class $\mathcal H$.
On the positive side, given a pattern graph~$H\in\mathcal H$ whose smallest vertex-cover has size $\vc(H)$ and an $n$-vertex host graph~$G$, there are known $O(n^{\vc(H)+1})$ time  algorithms~\cite{DBLP:journals/siamcomp/WilliamsW13,DBLP:journals/siamdm/KowalukLL13,DBLP:conf/focs/CurticapeanM14} to count subgraphs of $G$ that are isomorphic to $H$:
First, find a minimum vertex-cover~$C$ of~$H$ using exhaustive search. Then, iterate over all possible embeddings~$f$ of~$H[C]$ into~$G$ and count the possible extensions of $G[f(C)]$ to a full copy of~$H$.
Complementing this algorithm, an almost matching running time lower bound of $n^{\Omega(\vc(H) / \log \vc(H))}$ under the exponential-time hypothesis~($\ETH$) is also known~\cite{DBLP:conf/stoc/CurticapeanDM17}.
Thus, assuming $\ETH$ or $\FPT \neq \sharpW$, the problem $\#\Sub(\mathcal H)$ of counting subgraphs from a fixed class~$\mathcal H$ is polynomial-time solvable if and only if the vertex-cover numbers (or equivalently, the maximum matching sizes) of the graphs in~$\mathcal H$ are bounded by a constant.
The rightmost column of \cref{fig: overview} visualizes this situation.

Turning from counting to the problem $\Sub(\mathcal H)$ of \emph{detecting} subgraphs from fixed classes~$\mathcal H$, the picture is less clear.
Evidence points at three strata of complexity:
Define the \emph{matching-split number} of $H$ to be the minimum number of vertices whose deletion turns $H$ into a matching, that is, a graph of maximum degree $1$.
Jansen and Marx~\cite{JansenMarx} show that, if this number is bounded in a graph class~$\mathcal H$, then $\Sub(\mathcal H)$ is polynomial-time solvable.
For classes~$\mathcal H$ of bounded tree-width, it is known~\cite{DBLP:conf/wg/PlehnV90,DBLP:journals/algorithmica/AlonYZ97,DBLP:journals/jcss/FominLRSR12} that the problem $\Sub(\mathcal H)$ is fixed-parameter tractable when parameterized by~$|V(H)|$.
For pattern classes $\mathcal H$ of unbounded tree-width, it is conjectured that $\Sub(\mathcal H)$ is W[1]-hard---so far, this hardness has only been established for cliques, bicliques~\cite{DBLP:journals/jacm/Lin18}, grids~\cite{DBLP:conf/wg/ChenGL17}, and less natural graph classes.
The leftmost column of \cref{fig: overview} visualizes the situation.

We propose to study an intermediate setting between decision and counting, namely, counting subgraph patterns \emph{modulo fixed integers} $q \in \N$.
Modular counting has a tradition in classical complexity theory, where the complexity classes $\Mod_q \cc{P}$ for $q \in \N$ capture problems that ask to count accepting paths of polynomially time-bounded non-deterministic Turing machines modulo~$q$.
In particular, the class $\Mod_2 \cc{P}$ (better known as~$\cc{\oplus P}$) plays a central role in the proof of Toda's theorem~\cite{Toda}.
Several (partial) classification results for frameworks of modular counting problems are known; this includes homomorphisms to fixed graphs~\cite{DBLP:journals/toc/FabenJ15,DBLP:journals/toct/0001GR14,DBLP:journals/toct/0001GR16,DBLP:conf/mfcs/KazeminiaB19,doi:10.1137/1.9781611976465.137}, constraint satisfaction problems~\cite{DBLP:journals/corr/abs-0809-1836,DBLP:conf/stacs/GuoHLX11}, and Holant problems~\cite{DBLP:conf/focs/CurticapeanX15}.

\cref{fig: overview} summarizes our understanding.
If the vertex-cover number is bounded, the polynomial-time algorithms (regions 7 and 8) follow from the algorithm for $\#\Sub(\mathcal H)$ described above and require no further attention.
Our paper is concerned with the remaining regions~\mbox{1--6}.

As argued above, matchings play a central role in decision and counting, so it is natural that they reprise their role in modular subgraph counting:
On the positive side, there are known polynomial-time algorithms for counting matchings of a given size modulo fixed powers of two. (For bipartite graphs and counting modulo $2$, this essentially follows from the fact that determinant and permanent coincide modulo~$2$.)
On the negative side, if $q$ is not a power of two, counting matchings modulo $q$ is known to be $\Mod_p \cc{P}$-complete for any odd prime $p$ dividing $q$.
We establish a parameterized analogue of this fact: 
Let $\Mod_q\W$ be the class of parameterized problems that are fpt-reducible to counting $k$-cliques modulo~$q$.
We show that counting $k$-matchings (that is, sets of $k$ pairwise disjoint edges) in graphs modulo fixed odd primes $q \in \N$ is $\Mod_q\W$-hard.
In our proof, modular counting allows us to sidestep the algebraic machinery from previous works~\cite{DBLP:conf/focs/CurticapeanM14,DBLP:phd/dnb/Curticapean15,DBLP:conf/stoc/CurticapeanDM17}, resulting in a surprisingly simple and self-contained argument.
\begin{theorem}\label{thm: hardness-matchings}
  For any integer $q \in \N$ containing an odd prime factor $p$,
  counting $k$-matchings modulo $q$ is $\Mod_p\W$-hard under Turing fpt-reductions and admits no $n^{o(k/\log k)}$ time algorithm under $\ETH$.
\end{theorem}

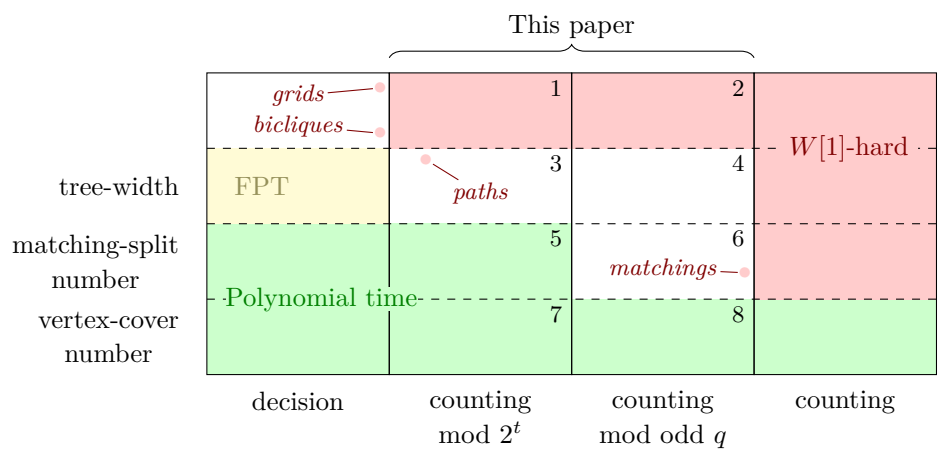
\begin{figure}[t]
  \begin{tikzpicture}[xscale=2.4]
    \draw [decorate,decoration={brace,amplitude=5pt}] (1,4.2) --  node [midway, yshift=1.2em] {This paper} (3,4.2);
    \fill [green!20] (0,0) rectangle (4,1);
    \fill [green!20] (0,1) rectangle (2,2);
    \fill [red!20] (1,3) rectangle (4,4);
    \fill [red!20] (3,1) rectangle (4,4);
    \fill [yellow!20] (0,2) rectangle (1,3);
    \draw (0,0) rectangle (1,4);
    \draw (1,0) rectangle (2,4);
    \draw (2,0) rectangle (3,4);
    \draw (3,0) rectangle (4,4);
    \draw [dashed] (0, 1) --(4,1);
    \draw [dashed] (0, 2) --(4,2);
    \draw [dashed] (0, 3) --(4,3);
    \node at (-.1,.5) [anchor = east, align=center] {vertex-cover\\number};
    \node at (-.1,1.5) [anchor = east, align=center] {matching-split\\number};
    \node at (-.1,2.5) [anchor = east] {tree-width};
    \node at (.5, -.1) [anchor = north] {decision};
    \node at (1.5, -.1) [anchor = north,align=center, text width = 2cm] {counting mod $2^t$};
    \node at (2.5, -.1) [anchor = north, align=center,text width = 2cm] {counting mod odd $q$};
    \node at (3.5, -.1) [anchor = north] {counting};
    \node (pt) at (1.5, 2.4) [inner sep=1pt, black!50!red] {\small\emph{paths}};
    \node (p) at (1.2, 2.85) [inner sep=0, red!20] {$\bullet$};
    \draw [black!50!red] (p) -- (pt);
    \node (mt) at (2.5, 1.4) [inner sep=1pt, black!50!red] {\small\emph{matchings}};
    \node (m) at (2.95, 1.35) [inner sep=0, red!20] {$\bullet$};
    \draw [black!50!red] (m) -- (mt);
    \node (ct) at (0.5, 3.7) [inner sep=1pt, black!50!red] {\small\emph{grids}};
    \node (c) at (0.95, 3.8) [inner sep=0, red!20] {$\bullet$};
    \draw [black!50!red] (c) -- (ct);
    \node (gt) at (0.5, 3.3) [inner sep=1pt, black!50!red] {\small\emph{bicliques}};
    \node (g) at (0.95, 3.2) [inner sep=0, red!20] {$\bullet$};
    \draw [black!50!red] (g) -- (gt);
    \node at (3.9, 3) [black!50!red, fill=red!20, anchor=east] {$W[1]$-hard};
    \node at (.1, 1) [black!50!green, fill=green!20, anchor=west, inner sep=0pt] {Polynomial time};
    \node at (.1, 2.5) [black!50!yellow, anchor=west] {FPT};
    \foreach \position/\text in { (1,0)/7, (2,0)/8, (1,1)/5, (2,1)/6, (1,2)/3, (2,2)/4, (1,3)/1, (2,3)/2 }
      \node at \position [xshift = 2.4cm, yshift = .8cm, anchor=east] {\small\text};
  \end{tikzpicture}

  \caption{\label{fig: overview}%
  An overview over known results and our new results.
  The columns correspond, from left to right, to the problem types $\Sub(\mathcal H)$, $\ModSub{\mathcal H}{2^t}$, $\ModSub{\mathcal H}{q}$ for $q\ne 2^t$, or $\#\Sub(\mathcal H)$; our results are depicted in the two middle columns.
  The rows correspond, from bottom to top, to requiring $\mathcal H$ to have bounded vertex-cover number, matching-split number, tree-width, or no requirement at all.
  The complexity along each row is monotone: By \cref{lem:decision-to-parity}, decision is no harder than modular counting, and modular counting trivially is no harder than counting.
  Our results are depicted in the middle two columns:
  Regions~1 and~2 are \cref{cor: hardness-tw}.
  Region~5 is \cref{thm: algorithm-split}.
  Regions~7 and~8 already follow from~\cite{DBLP:journals/siamcomp/WilliamsW13}.
  The point in region~6 is \cref{thm: hardness-matchings}, and the point in region~3 is \cref{thm: hardness-paths}.
  We view the hardness of these points as evidence to conjecture their enclosing regions to be hard, see \cref{conjecture}.
  }
\end{figure}

Known arguments from Ramsey theory (see~\cite[Section~5]{DBLP:journals/corr/CurticapeanM14}) extend \cref{thm: hardness-matchings} from matchings to $\ModSub{\mathcal H}{q}$ for any hereditary class~$\mathcal H$ of unbounded vertex-cover number.
This suggests that modular subgraph counting may only become tractable when the modulus is a power of two.
Indeed, we show that patterns of matching-split number $s$ can be counted modulo $q = 2^t$ in time $n^{O(t4^s)}$.
To prove this, we follow the general idea of the bounded vertex-cover number algorithm for $\Sub(\mathcal H)$ outlined before, and we reduce to counting matchings modulo powers of two. 
This however requires us to overcome technical complications to avoid unwanted cancellations.
Overall, we obtain:
\begin{theorem}\label{thm: algorithm-split}
  There is an algorithm that, given a graph $H$ of matching-split number $s \in \N$ and an $n$-vertex graph $G$,
  computes the number of $H$-isomorphic subgraphs of~$G$ modulo $2^t$ in time $n^{O(t4^s)}$.
\end{theorem}
We complement this result in two ways:
First, we observe that $\oplus\Sub(\mathcal H)$ is $\pW$-complete for pattern classes $\mathcal H$ of unbounded tree-width;
this follows directly from previous hardness proofs for $\#\Sub(\mathcal H)$.
More interestingly, we establish the $\pW$-completeness of counting $k$-paths modulo~$2$ in undirected graphs, thus solving an open problem from~\cite{DBLP:conf/icalp/BjorklundDH15}, where this problem was considered in the context of Hamiltonian cycle detection, following~\cite{BjorklundHusfeldt}.
\begin{theorem}\label{thm: hardness-paths}
  Counting $k$-paths modulo $2$ is $\pW$-complete.
\end{theorem}
This result adds to a rich range of previous work on the $k$-path problem, and is of interest outside our framework.
Bodlaender~\cite{BODLAENDER19931} and Monien~\cite{MONIEN1985239} showed that \emph{finding} a $k$-path is fixed-parameter tractable.
In contrast, Flum and Grohe~\cite{DBLP:journals/siamcomp/FlumG04} showed that \emph{exactly counting} $k$-paths is $\sharpW$-hard.
Nevertheless, Arvind and Raman~\cite{ArvindRaman} showed that \emph{approximately counting} $k$-paths, which corresponds to computing the most significant bit(s) of the number of~$k$-paths, is fixed-parameter tractable.
Our \cref{thm: hardness-paths} suggests that the \emph{least} significant bit of the number of $k$-paths is hard to compute.
This is surprising, because some of the most influential fpt-algorithms for finding a $k$-path work over characteristic~$2$,
based on the group algebra framework introduced by Koutis~\cite{Koutis08}.

\medskip

Let us conclude with a general remark on the techniques used in this paper:
Recent works successfully exploited a connection between subgraph counts and (linear combinations of) homomorphism counts to obtain algorithms and hardness results~\cite{DBLP:conf/stoc/CurticapeanDM17,DBLP:conf/esa/Roth17,DBLP:conf/mfcs/DorflerRSW19,DBLP:journals/algorithmica/RothS20,DBLP:conf/focs/Roth0W20}.
For example, the number of $k$-matchings in a graph $G$ is a linear combination of homomorphism counts from $f(k)$ fixed graphs.
Insights on the complexity of counting the homomorphisms occurring in this linear combination then lead to complexity results for counting $k$-matchings.
This connection however does not readily transfer to modular counting, as the relevant linear combinations (which involve rational coefficients) may be \emph{undefined} modulo~$p$. We therefore prove \cref{thm: hardness-matchings,thm: algorithm-split,thm: hardness-paths} using more combinatorial approaches.

\section{Preliminaries}\label{sec: preliminaries}
Unless otherwise stated, we consider finite, undirected, simple graphs without self-loops.

\paragraph*{Subgraph problems}

A \emph{homomorphism} from graph $H$ to graph $G$ is a mapping $\varphi\colon V(H)\to V(G)$ such that $\{\varphi(u),\varphi(v)\} \in E(G)$ for each $\{u,v\}\in E(H)$.
An \emph{embedding} is an injective homomorphism, and we let $\Emb(H,G)$ denote the set of embeddings from~$H$ to~$G$.
An \emph{isomorphism} is a bijective homomorphism,
and an \emph{automorphism} is an isomorphism from $H$ to itself. 
The set of all automorphisms of $H$ is called $\Aut(H)$, and forms a group when endowed with function composition~$\circ$.

We let $\Sub(H,G)$ be the set of all $H$-subgraphs of~$G$, that is, the set of all $H'$ with $V(H')\subseteq V(G)$ and $E(H')\subseteq E(G)$ such that $H'$ is isomorphic to~$H$.
This terminology fixes the possible confusion about isomorphic copies of subgraphs:
For example, there is exactly one $K_k$-subgraph in $K_k$,
but there are $k!$ embeddings.
The \emph{subgraph problem $\Sub$} is given a pair $(H,G)$ to decide whether $G$ has at least one $H$-subgraph.
The \emph{subgraph counting problem $\#\Sub$} is given a pair $(H,G)$ to determine the number of $H$-subgraphs in $G$.

For a graph class~$\classH$, we write $\#\Sub(\classH)$ for the restricted 
problem where the input $(H,G)$ is promised to satisfy $H\in \classH$.
For $q\in\Z_{\ge 2}$, the \emph{modular subgraph counting problem $\ModSub{\classH}{q}$} is the problem to compute the number of $H$-subgraphs modulo~$q$.
In the special case with~$q=2$, we write $\oplus\Sub$.

It will be useful to consider \emph{colorful} subgraph problems, where~$G$ is $H$-colored, that 
is, there is a given homomorphism~$c:V(G)\to V(H)$.
Due to the homomorphism property, we allow edges $\{u,v\}\in E(G)$ only if the 
corresponding colors satisfy $\{c(u), c(v)\} \in E(H)$.
A subgraph~$H'$ of an $H$-colored graph~$G$ is \emph{vertex-colorful} if $c$ is 
bijective on $V(H')$.
Let $\vcSub(H,G)$ be the set of vertex-colorful subgraphs~$H'$ for which $c$ is 
an isomorphism from $H'$ to $H$.
The corresponding computational problems are defined analogously to the 
uncolored case; the input consists of a graph~$G$ together with an 
$H$-coloring~$c$.

\paragraph*{Background from complexity theory}
A \emph{parameterized counting problem} is a pair $(f,\kappa)$ of functions $f,\kappa:\zo^\ast\to\N$ where $\kappa$ is computable.
A \emph{parsimonious fpt-reduction} from a parameterized counting problem~$(f,\kappa)$ to a parameterized counting problem~$(g,\iota)$ is a function~$R$ with the following properties:
(i) $f(x)=g(R(x))$ for all $x\in\zo^\ast$,
(ii) $\iota(R(x))$ is bounded by a computable function in~$\kappa(x)$, and
(iii) the reduction is computable in time $h(\kappa(x))\poly(\abs{x})$ for some computable function~$h$.
A~\emph{Turing fpt-reduction} may query the oracle multiple times for instances whose parameter is bounded by a function of the input parameter, and combine the query answers in fpt-time to produce the correct output.
Moreover, reductions can also be \emph{randomized}, in which case we require that their error probability is bounded by a small constant.

The \emph{exponential-time hypothesis} (ETH) postulates the existence of some~$\varepsilon>0$ such that no algorithm solves $n$-variable~$3$-CNF formulas in time~$O(2^{\varepsilon n})$. We write for short that $3$-CNF-SAT does not have $2^{o(n)}$-time algorithms, and we also disallow bounded-error randomized algorithms.

\paragraph*{Modular counting}
For our purposes, we define the class $\Mod_q\W$ as the class of all parameterized problems $(f,\kappa)$ with $f\colon \Sigma^\ast\to\{0,\dots,q-1\}$ such that $(f,\kappa)$ has a parsimonious fpt-reduction to the problem of counting $k$-cliques modulo~$q$.
For $q=2$, it was shown in~\cite{DBLP:conf/icalp/BjorklundDH15} that 
all problems in $\W$ admit randomized fpt-reductions to problems in $\pW$.
Another result~\cite[Lemma~2.1]{DBLP:conf/soda/WilliamsWWY15} yields the corresponding generalization for all $q>2$.
We use the following analogous proposition for the vertex-colorful subgraph problem, proven in the appendix.

\begin{lemma}\label{lem:decision-to-parity}%
  For any integer $q \geq 2$,
  there is a randomized Turing fpt-reduction from the problem~$\vcSub$ to the problem~$\#\vcSub\bmod q$.
  On input $(H,G)$, the reduction only queries instances with the same pattern $H$.
\end{lemma}

Our work relies on the following hardness result for parameterized modular subgraph counting, which follows easily from known results on the colorful subgraph decision problem~\cite{DBLP:journals/tcs/DalmauJ04,DBLP:journals/toc/Marx10}. See the appendix for a proof.

\begin{lemma}\label{cor: hardness-tw}
  Let $\classH$ be a graph family of unbounded tree-width and let~$q$ be an integer with $q\ge 2$.
  Then $\#\vcSub(\classH)\bmod q$ parameterized by $k=\abs{E(H)}$ is $\Mod_q\W$-hard under parsimonious fpt-reductions.
  Moreover, if ETH is true, then the problem does not have an algorithm running
  in time $n^{o(k/\log k)}$, where~$n=\abs{V(G)}$. 
\end{lemma}

\section{Hardness of counting $k$-matchings}\label{sec: hardness-matchings}

In this section, we prove \cref{thm: hardness-matchings}.
We first establish $\Mod_q\W$-hardness of the problem $\#\pp{ColMatch}\bmod q$ for odd $q \geq 3$:
Given a graph $G$ with an edge-coloring $c: E(G) \to \mathcal C$ for some set of colors $\mathcal C$ with $\abs{\mathcal C}=k$, 
this problem asks to count modulo $q$ the edge-colorful matchings in $G$.
These are the matchings that use each color in $\mathcal C$ exactly once. 

\begin{lemma}\label{lem:reduction-to-colmatch}
  For any fixed integer $p$ with odd prime factor $q$,
  the problem ${\#\pp{ColMatch}\bmod p}$ is $\Mod_q \cc{W[1]}$-hard under parsimonious fpt-reductions and has no $n^{o(k/\log k)}$ time algorithm under $\ETH$.
\end{lemma}
\begin{proof}
  The class $\mathcal H$ of all $3$-regular graphs has unbounded tree-width, and hence by \cref{cor: hardness-tw}, the problem $\#\vcSub(\classH)\bmod q$ is $\Mod_q \cc{W[1]}$-hard and hard under $\ETH$.
  We reduce it to $\#\pp{ColMatch}\bmod q$, implying the hardness of $\#\pp{ColMatch}\bmod p$.
  Let $H\in\classH$ and~$G$ be the input for the reduction with $k=\abs{V(H)}$ and $H$-colored $G$, 
  and let $\setc{V_a}{a \in V(H)}$ be the color classes of $G$,
  with edge-sets $E_{a,b}(G)$ for $ab \in E(H)$.
  Using the gadgets from \cref{fig:matchings mod p}, we construct a graph~$G'$:

  \begin{enumerate}
    \item Each vertex $u \in V(G)$ is replaced by three vertices $u_1$, $u_2$, and $u_3$.
    We insert a \emph{consistency gadget} $Q_u$ at these vertices by adding $q$ gadget vertices, connecting $u_2$ and $u_3$ to all gadget vertices, and $u_1$ to the first $(q+1)/2$ gadget vertices. 
    For $S \subseteq \set{u_1,u_2,u_3}$, let $m_S$ count the matchings in $Q_u$ that match precisely $S$;
    it can be checked that
    \begin{equation}
    \label{eq: consistency-gadget-extensions}
      m_S 
      \equiv_q 
      \begin{cases}
        1 & \text{if $S$ is $\emptyset$ or $\set{u_1,u_2,u_3}$,} 
        \\
        0 & \text{if $S$ is $\set{u_2}$, $\set{u_3}$, or $\set{u_2,u_3}$.}
      \end{cases}
    \end{equation}
    We explicitly ignore the other three cases for $S$, as they will not be relevant.
    
  \item For $\{a,b\} \in E(H)$, suppose that $a$ is the $j$th vertex incident to $b$, and that $b$ is the $i$th vertex incident to~$a$. 
    For each edge $\{u,v\} \in E_{a,b}(G)$ with $u\in V_a$ and $v\in V_b$, we insert an AND-gadget $A_{uv}$ at $\set{u_i,v_j}$. Then for any set $S\subsetneq \set{u_i,v_j}$, the number of edges in $A_{uv} - S$ is divisible by $q$, whereas $A_{uv} - \set{u_i,v_j}$ has exactly one edge.
    \item The edge-colors of $G'$ are defined as follows:
    For~$u\in V_a$ with $a \in V(H)$, we assign color $(\mathtt{CONS},a,i)$ to all edges of $Q_u$ incident to vertex~$u_i$, for $i\in\set{1,2,3}$.
    For each $ab \in E(H)$, we assign color $(\mathtt{AND},ab)$ to all edges in AND-gadgets between edges in $E_{a,b}(G)$.
    Overall, we have $k'=3k+3k/2$ colors.
  \end{enumerate}

  \begin{figure}[tp]
  	
  	    \centering
  	\begin{subfigure}[t]{0.4\textwidth}
  		\centering
\begingroup%
  \makeatletter%
  \providecommand\color[2][]{%
    \errmessage{(Inkscape) Color is used for the text in Inkscape, but the package 'color.sty' is not loaded}%
    \renewcommand\color[2][]{}%
  }%
  \providecommand\transparent[1]{%
    \errmessage{(Inkscape) Transparency is used (non-zero) for the text in Inkscape, but the package 'transparent.sty' is not loaded}%
    \renewcommand\transparent[1]{}%
  }%
  \providecommand\rotatebox[2]{#2}%
  \newcommand*\fsize{\dimexpr\f@size pt\relax}%
  \newcommand*\lineheight[1]{\fontsize{\fsize}{#1\fsize}\selectfont}%
  \ifx\svgwidth\undefined%
    \setlength{\unitlength}{112.5bp}%
    \ifx\svgscale\undefined%
      \relax%
    \else%
      \setlength{\unitlength}{\unitlength * \real{\svgscale}}%
    \fi%
  \else%
    \setlength{\unitlength}{\svgwidth}%
  \fi%
  \global\let\svgwidth\undefined%
  \global\let\svgscale\undefined%
  \makeatother%
  \begin{picture}(1,0.54956523)%
    \lineheight{1}%
    \setlength\tabcolsep{0pt}%
    \put(0,0){\includegraphics[width=\unitlength,page=1]{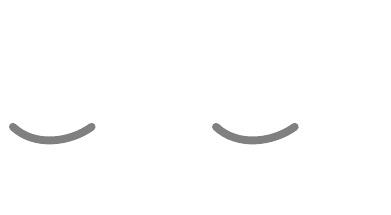}}%
    \put(0.77740926,0.17569708){\color[rgb]{0.50196078,0.50196078,0.50196078}\makebox(0,0)[lt]{\lineheight{1.25}\smash{\begin{tabular}[t]{l}$q-2$\end{tabular}}}}%
    \put(0,0){\includegraphics[width=\unitlength,page=2]{gadgets_and.pdf}}%
    \put(-0.00316839,0.3490201){\makebox(0,0)[lt]{\lineheight{1.25}\smash{\begin{tabular}[t]{l}$u_i$\end{tabular}}}}%
    \put(0.71300971,0.34705992){\makebox(0,0)[lt]{\lineheight{1.25}\smash{\begin{tabular}[t]{l}$v_j$\end{tabular}}}}%
    \put(0.24581462,0.17569708){\color[rgb]{0.50196078,0.50196078,0.50196078}\makebox(0,0)[lt]{\lineheight{1.25}\smash{\begin{tabular}[t]{l}$q-2$\end{tabular}}}}%
  \end{picture}%
\endgroup%

  		\subcaption{The AND-gadget, shown here for $q=5$.
  			In general, the upper $u_i,v_j$-path always has $3$ edges; both \emph{external vertices} $u_i$ and~$v_j$ have $q-2$ neighboring leaves.
  			If exactly~$0$ or~$1$ external vertices are
  			removed, this graph has $0$ edges modulo~$q$; if both vertices are removed, 
  			the graph has $1$ edge.}
  	\end{subfigure}%
  	\hspace{1cm}
  	\begin{subfigure}[t]{0.4\textwidth}
  		\centering
\begingroup%
  \makeatletter%
  \providecommand\color[2][]{%
    \errmessage{(Inkscape) Color is used for the text in Inkscape, but the package 'color.sty' is not loaded}%
    \renewcommand\color[2][]{}%
  }%
  \providecommand\transparent[1]{%
    \errmessage{(Inkscape) Transparency is used (non-zero) for the text in Inkscape, but the package 'transparent.sty' is not loaded}%
    \renewcommand\transparent[1]{}%
  }%
  \providecommand\rotatebox[2]{#2}%
  \newcommand*\fsize{\dimexpr\f@size pt\relax}%
  \newcommand*\lineheight[1]{\fontsize{\fsize}{#1\fsize}\selectfont}%
  \ifx\svgwidth\undefined%
    \setlength{\unitlength}{150bp}%
    \ifx\svgscale\undefined%
      \relax%
    \else%
      \setlength{\unitlength}{\unitlength * \real{\svgscale}}%
    \fi%
  \else%
    \setlength{\unitlength}{\svgwidth}%
  \fi%
  \global\let\svgwidth\undefined%
  \global\let\svgscale\undefined%
  \makeatother%
  \begin{picture}(1,0.41217392)%
    \lineheight{1}%
    \setlength\tabcolsep{0pt}%
    \put(0,0){\includegraphics[width=\unitlength,page=1]{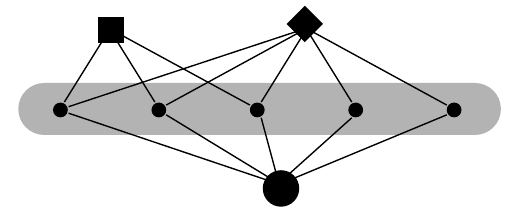}}%
    \put(0.10385102,0.35005812){\makebox(0,0)[lt]{\lineheight{1.25}\smash{\begin{tabular}[t]{l}$u_1$\end{tabular}}}}%
    \put(0.62917816,0.36005812){\makebox(0,0)[lt]{\lineheight{1.25}\smash{\begin{tabular}[t]{l}$u_2$\end{tabular}}}}%
    \put(0.58905548,0.01324864){\makebox(0,0)[lt]{\lineheight{1.25}\smash{\begin{tabular}[t]{l}$u_3$\end{tabular}}}}%
    \put(0.9036915,0.280216){\color[rgb]{0.50196078,0.50196078,0.50196078}\makebox(0,0)[lt]{\lineheight{1.25}\smash{\begin{tabular}[t]{l}$q$\end{tabular}}}}%
  \end{picture}%
\endgroup%

  		\subcaption{The consistency gadget contains $q$ gadget vertices, shown here for $q=5$. The number of matchings of size~$0$ and $3$ equals~$1$ modulo~$q$, and if $u_1$ is deleted, the number of non-empty edge-colorful matchings equals~$0$ modulo $q$.}
  	\end{subfigure}
  	\caption{\label{fig:matchings mod p}%
    The two gadgets used in the proof of \cref{lem:reduction-to-colmatch}.}
  \end{figure}
 
  Every $H$-copy $F$ in $G$ induces a set $\mathcal{M}_F$ of colorful matchings in $G'$. We describe this set in the following, show that $|\mathcal{M}_F| \equiv_q 1$, and that $\mathcal{M}_F$ and  $\mathcal{M}_{F'}$ are disjoint for $F \neq F'$.
  \begin{itemize}
    \item For each $v \in V(F)$, match all of $\set{v_1,v_2,v_3}$ within $Q_v$. 
    For fixed $v$, the number of possible matchings in $Q_v$ is $m_{\set{v_1,v_2,v_3}} \equiv_q 1$ by~\eqref{eq: consistency-gadget-extensions}.
    Let $\mathcal{Q}_F$ denote the set of all matchings that can be obtained by the previous step.
    Since matchings can be chosen independently for distinct $Q_v$, we obtain $|\mathcal{Q}_F| \equiv_q 1^{|V(F)|} \equiv_q 1$.
  \item Any $M \in \mathcal{Q}_F$ can be extended to several colorful matchings by choosing one edge from each color $(\mathtt{AND},ab)$ for $\{a,b\} \in E(H)$.
    For each $\{u,v\} \in E(F)$, the AND-gadget $A_{uv}$ has exactly one such edge,
    while the other AND-gadgets of color $(\mathtt{AND},ab)$ have $0$ such edges modulo $q$.
    Hence, the number edges of color $(\mathtt{AND},ab)$ that can extend $M$ is $1$ modulo $q$.
    This implies that the overall number $r_M$ of matchings extending $M$ into a colorful matching is also $r_M \equiv_q 1^{|E(F)|} \equiv_q 1$.
  \end{itemize}
  Overall, every $H$-copy $F$ induces $\sum_{M \in \mathcal{Q}_F} r_M \,\equiv_q\, \sum_{M \in \mathcal{M}_F} 1 \,\equiv_q\, 1$ colorful matchings, so we indeed have $|\mathcal{M}_F| \equiv_q 1$. We also observe from the construction that $\mathcal{M}_F \cap \mathcal{M}_{F'} = \emptyset$ for distinct $H$-copies $F$ and $F'$. In the appendix, we use properties of the gadgets to prove that colorful matchings $M \notin \bigcup_{F} \mathcal{M}_F$ cancel modulo $q$.
\begin{claim}\label{claim: irrelevant-matchings}
  The number of colorful matchings $M$ that are not contained in $\mathcal{M}_F$ for any $H$-copy $F$ is divisible by $q$.
\end{claim}
Overall, we have shown that the number of $H$-copies in $G$ and the number of colorful matchings in $G'$ agree modulo $q$. As $G'$ can be computed in polynomial time and the parameter is increased only by a constant factor, the claimed hardness results follow.
\end{proof}

To prove \cref{thm: hardness-matchings}, it suffices to give an fpt-reduction from ${\#\pp{ColMatch}\bmod q}$ to counting $k$-matchings modulo $q$. 
This is achieved by a standard inclusion-exclusion argument that can be found in the appendix.

\section{Counting matching-splittable subgraphs modulo $2^t$}\label{sec: algorithm}

In this section, we prove \cref{thm: algorithm-split} by describing an $n^{O(t4^s)}$-time algorithm for counting modulo $2^t$ the subgraphs of matching-split number $s$.
Our algorithm builds upon known algorithms for the decision and counting versions of subgraph problems;
we first review their underlying ideas and sketch our algorithm for \cref{thm: algorithm-split}.

\subsubsection*{Counting subgraphs of bounded vertex-cover number.}
The basic structure of our algorithm is similar to a known $O(n^{s+1})$ time algorithm~\cite{DBLP:journals/siamcomp/WilliamsW13,DBLP:journals/siamdm/KowalukLL13,DBLP:conf/focs/CurticapeanM14} for counting embeddings from~$H$ to~$G$ if $H$ has a vertex-cover $S\subseteq V(H)$ of size $s\in\N$. Counting embeddings is sufficient for counting subgraph copies, as we can first compute the number $\#\Aut(H)$ of automorphisms on~$H$ as $\#\Emb(H,H)$, and then use
\begin{equation}\label{eq: sub-emb-aut}%
\#\Sub(H,G) = \frac{\#\Emb(H,G)}{\#\Aut(H)}\,.
\end{equation}

Given $(H,G)$ with $h=\abs{V(H)}$ and $n=\abs{V(G)}$, the algorithm for computing $\#\Emb(H,G)$ first finds a minimum vertex-cover~$S$ of~$H$ in time~$h^{O(s)}$; then $I:=V(H)\setminus S$ is an independent set.
Then the algorithm enumerates all partial embeddings~$f$ from $H[S]$ to $G$, which takes time at most~$n^{O(s)}$.
Finally, for each~$f$, it remains to count all functions $g\colon I \to V(G)$ that extend~$f$ to a full embedding from $H$ to~$G$.
We observe that $g$ extends $f$ to a full embedding if and only if every vertex $u\in I$ maps via~$g$ to a vertex~$v=g(u)\in V(G)\setminus f(S)$ that satisfies the \emph{neighborhood constraint} $N_G(v)\cap f(S)\supseteq f(N_H(u))$.
Counting functions~$g$ with this property can be achieved (in a not completely obvious way) with dynamic programming;
we only need to know the number of vertices $v\in V(G)\setminus f(S)$ that have a specific neighborhood $N_G(v)\cap f(S)$, and for each~$f$, there are at most $2^s$ different possible such neighborhoods.
Overall, in~$n^{O(s)}$ time, we can compute the number $\#\Emb(H,G)$.

\subsubsection*{Detecting subgraphs of bounded matching-split number.}
Jansen and Marx~\cite{JansenMarx} extend the above approach and obtain an $n^{O(s)}$ time algorithm for the \emph{decision} problem $\Sub(H,G)$ 
when~$H$ has matching-split number~$s$.
In this case, we consider a \emph{splitting set} $S$ of size~$s$ instead of a vertex-cover, that is, the graph $M=H-S$ may have isolated edges besides isolated vertices.
Now the idea is to not only classify the vertices $v\in V(G)\setminus f(S)$ by their neighborhoods $N_v=N_G(v)\cap f(S)$, but to also classify the edges $\set{u,v}\in E(G-f(S))$ by their neighborhoods $\set{N_u,N_v}$.
It then remains to find a matching in~$G-f(S)$ that has as many isolated vertices and isolated edges as $H-S$, such that these vertices and edges satisfy the neighborhood constraints in~$f(S)$.
Jansen and Marx achieve this by reduction to a colored matching problem.

\subsubsection*{Our algorithm.}
In our algorithm for \cref{thm: algorithm-split}, we need to overcome two challenges:
\begin{enumerate}[(a)]
  \item\label{challenge-modulo} Since counting embeddings is algorithmically more straight-forward than counting subgraphs, we would like to count embeddings and divide by the number of automorphisms $\#\Aut(H)$ as in \cref{eq: sub-emb-aut}. 
  However, since we are counting modulo~$2^t$, the number $\#\Aut(H)\bmod 2^t$ may be~$0$, and so the division in \cref{eq: sub-emb-aut} is impossible. (In fact, even even numbers $\#\Aut(H)$ have no inverse modulo $2^t$.)
  \item\label{challenge-count} When mimicking Jansen and Marx's detection algorithm, we cannot just \emph{count} the relevant matchings in $G-f(S)$, since counting perfect matchings is~$\sharpP$-hard.
\end{enumerate}
Most of our effort focuses on overcoming~\eqref{challenge-modulo}:
In \cref{sub: rigidize-split}, we show that every graph~$H$ of matching-split number $s$ has a splitting set $R$ of size $O(s^2)$ that remains rigid under automorphisms, i.e., any automorphism $f$ of $H$ must 
satisfy $f(R) = R$.
In \cref{sub: fixed split}, we show how to compute $\#\Sub(H,G)$ if 
such a rigid splitting set~$R$ for~$H$ is given.
Rather than counting $H$-embeddings and attempting a division by $\#\Aut(H)$, we use the rigidity of $R$ to keep track of the automorphisms of~$H$ in a more explicit way.

To overcome~\eqref{challenge-count}, we use a determinant-based algorithm~\cite{HiraiNamba} to compute the Hafnian over a polynomial ring modulo~$2^t$ in \cref{sec:Hafnian}.
We then reduce our constrained matching counting problem to computing such Hafnians in \cref{sec:color-demands}.

\subsection{Rigidizing the splitting set}%
\label{sub: rigidize-split}

Let $H$ be a graph with a splitting set $S$ of size $s$, and let 
$M=H-S$ be the remaining graph of maximum degree $1$; we speak of $M$ as a matching, even though it may contain isolated vertices.
An automorphism~$f$ of $H$ may map a vertex $v\in S$ in the splitting set to $f(v) \notin S$.
We show that if a splitting set of size $s$ exists then there is also a \emph{rigid} splitting set~$R$ of size~$O(s^2)$, i.e., 
such that every $f \in \Aut(H)$ satisfies $f(R) = R$.
In fact, the following algorithm can find such a set~$R$.

\begin{algorithm}{Rigidize}{$H$}{Given a graph $H$ of matching-split number $s$, this algorithm computes a rigid splitting set $R\subseteq V(H)$ of size $O(s^2)$.}
\item[R1] (Find small splitting set.)
  Using brute-force, compute a set~$S\subseteq V(H)$ of size~$s$ such that $H-S$ is a matching.
\item[R2] (Extend it to neighbors of low-degree vertices.)
  Let $D\subseteq V(H)$ be the set of all vertices whose degree in~$H$ is at most~$s+1$.
  Set $T:=S$.
  While there is an edge $\{u, v\}$ with $u\in T\cap D$ and $v\in\overline T$, add $v$ to~$T$.
\item[R3] (Refine it.)
  Set $R:=T$.
  For each component~$C$ of~$H[T\cap D]$ with at most two vertices, 
  remove~$V(C)$ from~$R$.
\end{algorithm}
The following lemma captures useful properties of Rigidize. See the appendix for a proof.

\begin{lemma}\label{lem:rigidize}
  The algorithm Rigidize runs in time $h^{O(s)}$ where $h=\abs{V(H)}$, and the output set $R\subseteq V(H)$ has the properties that $\abs{R}\le O(s^2)$, that $H-R$ is a matching, and that every $f \in \Aut(H)$ satisfies $f(R)=R$.
\end{lemma}

\subsection{Counting subgraphs with rigid splitting sets}\label{sub: fixed split}

We use the rigid splitting set~$R$ from \cref{lem:rigidize} to 
compute the number of times~$H$ occurs as a subgraph modulo a power of two.
As a subroutine, we use an algorithm for counting colored matchings modulo a power of two in a setting involving particular ``color demands''.
\begin{definition}
Let~$G$ be a graph, let~$C$ be a finite set of colors, and let $c\colon V(G)\cup E(G)\to 2^C$ be a function that labels each vertex and edge with a subset of~$C$.
For any matching~$M$, let $I(M)$ be the set of its isolated vertices.
  For a coloring $c_M:I(M)\cup E(M)\to C$, the colored matching $(M,c_M)$ is \emph{permissible} if $c_M(t)\in c(t)$ holds for all~$t\in I(M)\cup E(M)$.

\emph{Color demands} are functions~$D_I,D_E\colon C \to \N$.
The pair $(M,c_M)$ \emph{satisfies the demands~$D_I,D_E$} if, for each~$i\in C$, the graph~$M$ contains exactly $D_I(i)$ isolated vertices~$v$ with $c_M(v)=i$ and exactly $D_E(i)$ edges with $c_M(e)=i$.
Let $\mathcal M(G,c,D_I,D_E)$ be the set of all permissible matchings~$(M,c_M)$ that satisfy the demand~$D$.
\end{definition}

As shown in the appendix, we obtain the following algorithm as a corollary to Hirai and Namba's algorithm~\cite{HiraiNamba} for computing the Hafnian over polynomial rings modulo~$2^t$.
\begin{lemma}\label{lem:color-demands}
  Given a graph~$G$, permissible colors~$c\colon V(G)\cup E(G)\to 2^C$, color demands $D_I,D_E\colon C\to\N$, and $t\in\N_{\ge1}$, there is an algorithm that computes the number $\abs{\mathcal M(G,c,D_I,D_E)}\bmod 2^t$ in time $n^{O(t\abs{C})}$.
\end{lemma}

Before we state the main algorithm, we introduce some basic group-theoretic notation.
Let~$R$ be a splitting set of~$H$ that satisfies $f(R)=R$ for all $f\in\Aut(H)$.
Let $G$ be a graph and let $S\subseteq V(G)$ be a set with $\abs{S}=\abs{R}$.
For an embedding $\sigma\in\Emb(H[R],G[S])$ and an 
automorphism~$\varphi\in\Aut(H)$, we note that the function 
$\sigma\circ(\varphi\restrict[R])$ is again an embedding in $\Emb(H[R],G[S])$.
Indeed, we view this operation as a right-action of the group $\Aut(H)$ on the 
set $\Emb(H[R],G[S])$.
We call two embeddings $\sigma,\sigma'\in\Emb(H[R],G[S])$ \emph{equivalent} if 
there exists~$\varphi\in\Aut(H)$ such that
$\sigma'=\sigma\circ(\varphi\restrict[R])$; this clearly defines an equivalence 
relation.
The equivalence class $\sigma\Aut(H)$ is called the \emph{orbit} of $\sigma$ 
under $\Aut(H)$.
All orbits have the same size.
Let $E_S$ be a set of representatives for each orbit, that is, a maximal set of 
mutually non-equivalent embeddings in~$\Emb(H[R],G[S])$.

We are ready to state the modular counting algorithm for $s$-matching-splittable subgraphs.

\begin{algorithm}{ModCount}{$H,G,t$}{Given an $s$-matching-splittable graph~$H$, a 
    host graph~$G$, and an integer~$t\ge 2$, this algorithm computes the number 
    $\#\Sub(H,G)\bmod 2^t$.}
\item[C1] (Compute rigid splitting set.)
  Call Rigidize$(H)$ to compute the set~$R$.
\item[C2] (Reduce to counting colored matchings.)
  For each $S\subseteq V(G)$ with $\abs{S}=\abs{R}$ (that is, a possible image of $R$)
  and each representative embedding $\sigma\in E_S$ from $H[R]$ to $G[S]$,
  we construct an instance~$(G-S,c_\sigma,D_I,D_E)$ of colored matching with demands and
  use \cref{lem:color-demands} to obtain the number $\abs{\mathcal M(G-S,c_\sigma,D_I,D_E)}\bmod 2^t$:
  \begin{description}
    \item[a.]
      (Set permitted colors.)
      Let~$C=2^R\cup \binom{2^R}{1}\cup\binom{2^R}{2}$.
      For each vertex ${v\in V(G)\setminus S}$, 
      let $N_v\subseteq R$ be the vertices of $R$ that hit $N_G(v)\cap S$ under $\sigma$, that is,
      $N_v=\sigma^{-1}(N_G(v)\cap S)$.
      Define~$c_\sigma(v)=\setc{N}{N\subseteq N_v}$.
      Moreover, for each $\set{u,v}\in E(G-S)$, define $c_\sigma(\set{u,v})=\setc{\set{N,N'}}{N\subseteq N_u,\; N'\subseteq N_v}$.
    \item[b.]
      (Make demands.)
      The demands $D_I,D_E:C\to\N$ depend only on~$H$ and~$R$.
      For each $N\subseteq R$, we let 
      $D_I(N)$ be the number of isolated vertices~$v$ in $H-R$ whose neighborhood in~$H$ 
      satisfies $N_H(v)\cap R=N$.
      Moreover, for all $N,N'\subseteq R$, we let 
      $D(\set{N,N'})$ be the number of edges~$\set{u,v}\in E(H-R)$ with $\set{N_H(u)\cap R, N_H(v)\cap R}=\set{N,N'}$.
  \end{description}
\item[C3] (Sum up.)
  Output the sum modulo $2^t$ of all integers returned by the queries in C2.
\end{algorithm}

  In the appendix, we prove that ModCount satisfies the properties stated in \cref{thm: algorithm-split}.

\section{Hardness of counting paths modulo two}\label{sec: hardness-paths}

In this section, we prove \cref{thm: hardness-paths}, that counting $k$-paths modulo $2$ is $\pW$-hard.
We first formally introduce this and some intermediate problems.

The \emph{length} of a path is the number of its edges.
For a graph~$G$ and vertices~$s,t\in V(G)$, an \emph{$s,t$-path} is a simple path from~$s$ to~$t$.
For a computable, strictly increasing function~$f\colon\N\to\N$, we define \pp{$f$-Flexible Path} to be the problem that is given~$(G,s,t,k)$ to decide whether there exists any~$s,t$-path in~$G$ whose length~$\ell$ satisfies $k\le\ell\le f(k)$.
When $\id$ denotes the identity function, then \pp{Path} (also known as \pp{$k$-Path} or \textsc{Longest Path}) is defined as \pp{$\id$-Flexible Path}.
We similarly define \pp{Directed $f$-Flexible Path} and \pp{Directed Path} for directed graphs, and we define the counting and parity versions of these problems in the canonical manner.

We start our reduction at the vertex-colorful subgraph problem $\oplus\vcSub(\mathcal H)$ for a class $\mathcal H$ of unbounded tree-width, which is $\pW$-hard by \cref{cor: hardness-tw}.
The core part of the reduction is in \cref{lem:vcsubs-to-dipaths}, where we reduce to counting paths of somewhat flexible length in a directed graph (modulo~$2$).
From there, we reduce to the familiar $k$-path problem in undirected graphs using standard tricks.

\begin{lemma}\label{lem:vcsubs-to-dipaths}%
  For any class $\mathcal H$ of connected $4$-regular graphs without non-trivial automorphisms,
  there is a computable strictly increasing function~$f$
  such that $\oplus\vcSub(\mathcal H)$ admits a parsimonious polynomial-time fpt-reduction to \pp{$\oplus$Directed $f$-Flexible Path}.
\end{lemma}

\begin{proof}
  Let $H\in\mathcal H$ and $G$ be given as input, 
  where $G$ is given with color classes $V_u$ for~$u\in V(H)$.
  Let $k=\abs{E(H)}$.
  Since $H$ is connected and $4$-regular, it has an Eulerian tour 
  $u_0,u_1,\dots,u_{k-1},u_{k}$ such that $E(H)=\setc{ \set{u_i,u_{i+1}} }{ i\in\set{0,\dots,k-1} }$.
  Every vertex of~$H$ appears exactly twice on the tour, except 
  for $u_0=u_{k}$, which appears three times.
  We construct a graph $G'$ that visits every color class of~$G$ two (or three) 
  times according to the Eulerian tour in~$H$.

  Before we give a formal construction, we give an overview.
  The graph~$G'$ will essentially be a sequence of bipartite 
  graphs~$B_0,\dots,B_\ell$ whose edges are all directed from left to right.
  For a bipartite graph $B$, we write $L(B)$ and $R(B)$ for 
  its left and right part, respectively.
  We will have $R(B_j)=L(B_{j+1})$ for all $j\in\set{0,\dots,\ell}$.
  Each $B_j$ will either be a perfect matching~$M_i$ or a graph $G_i$ that is a 
  copy of $G[V_{u_{i-1}}\cup V_{u_{i}}]$. (Note that $G_i$ is indeed bipartite, since $G$ is $H$-colored and $H$ contains no self-loops.)
  Pictorially, the sequence of bipartite graphs is $M_0 G_1 M_1 \dots M_{k-1} 
  G_k M_k M_{k+1}$.
  We also add some additional gadget edges between non-neighboring 
  layers.
  
  We now describe the construction of $G'$ in detail:
  \begin{enumerate}
    \item
      \textbf{Graph edges.}
      For each $i\in\set{1,\dots,k}$, let $G_i$ be a fresh copy of 
      $G[V_{u_{i-1}} \cup V_{u_{i}}]$, renamed so that $L(G_i)=\set{i}\times 
      V_{u_{i-1}}$ and $R(G_i)=\set{i}\times V_{u_i}$.

    \item
      \textbf{Matching edges.}
      For each $i\in\set{1,\dots,k-1}$, we connect $G_i$ with $G_{i+1}$ by 
      adding the canonical perfect matching~$M_i$ from $L(M_i)\doteq 
      R(G_i)=\set{i}\times V_{u_i}$ to $R(M_i)\doteq L(G_{i+1})=\set{i+1}\times 
      V_{u_i}$.

      At the last graph~$G_{k}$, we add a dangling matching $M_{k}$ with 
      $L(M_k)=R(G_k)=\set{k}\times V_{u_k}$ and $R(M_k)=\set{k+1}\times 
      V_{u_k}$.
      We add a further dangling matching $M_{k+1}$ with $L(M_{k+1})=R(M_k)$ and 
      $R(M_{k+1})=\set{k+1}\times V_{u_k}$.
      And a final dangling matching~$M_0$ to the left of $G_1$ so that 
      $R(M_0)=L(G_1)$ and $L(M_0)=\set{0}\times V_{u_0}$.

    \item
      \textbf{Gadget edges.}
      For all $i,j\in\set{1,\dots,k}$ with $i < j$ and $u_i = u_j$, note that 
      $L(M_i)=\set{i}\times V_{u_i}$ and $L(M_j)=\set{j}\times V_{u_i}$.
      We add the canonical \emph{bidirected} matching between $L(M_i)$ and 
      $L(M_j)$.
      Similarly, we add the canonical bidirected matching between $R(M_i)$ and 
      $R(M_j)$.
      We also add those matchings between $L(M_0)$ and $L(M_{k+1})$ and between 
      $R(M_0)$ and $R(M_{k+1})$.

    \item
      \textbf{Source/sink.}
      Let $s$ be a new vertex and add all edges $(s,v)$ for $v\in L(M_0)$.
      Let $t$ be a new vertex and add all edges $(v',t)$ for $v'\in R(M_{k+1})$.

    \item
      \textbf{Parameters.}
      Finally, we set $k' = 2k + 4$ and $f(k') = k + 3(k+2) + 2$ so that we are 
      counting all $s,t$-paths whose length is between $k'$ and $f(k')$.
  \end{enumerate}
  This finishes the construction of~$G'$.
  Note that $G'\setminus\set{s,t}$ is indeed a sequence $B_0\dots B_\ell$ of 
  bipartite graphs with some additional gadget edges, where $\ell=2k+2$.
  We define the \emph{$j$-th layer of $G'$} as the set $L_j=L(B_j)$ for 
  $j\le\ell$ and $L_{\ell+1}=R(B_\ell)$ for $j=\ell+1$.
  Recall that $L_j=R(B_{j-1})$ for $j>0$.

  \medskip
  We describe the \emph{canonical} solutions in the output of the reduction.

  Thus, let $H'$ be an $H$-subgraph of $G$ that is colorful.
  Recall that this means that $H$ is isomorphic to $H'$ and that the `coloring' homomorphism $c\colon V(G)\rightarrow V(H)$ is bijective on $H'$.
  Since $H'$ is finite, $c$ is in fact an isomorphism.
  (It must map non-edges to non-edges because $E(H)$ and $E(H')$ have the same size.)
  Let $\phi\colon H'\rightarrow H$ be the inverse of $c$ and note that $\phi$ is a homomorphism.

  We define the canonical $s,t$-path $p=p_{H'}$ in $G'$ corresponding to $H'$ by picking one 
  vertex from each layer.
  More precisely, $p$ is the unique $s,t$-path that does not use any gadget 
  edges and that satisfies the following property for all 
  $j\in\set{0,\dots,\ell+1}$: Setting $i$ and $u$ such that $L_j=\set{i}\times 
  V_u$, we have $V(p) = \set{(i,\phi(u))}$.
  Note that this determines a unique vertex set of~$p$.
  We claim that $V(p)$ indeed induces a path on the graph and matching edges, 
  which must be unique if it exists.

  To show that $p$ is a path, let $j\in\{0,\ldots, \ell\}$.
  We claim that the two vertices in $V(p)\cap V(B_j)$ are adjacent in $B_j$.
  If $B_j$ is one of the matching graphs, then $U_j=\set{i}\times V_u$ for some 
  $i$ and $u$, and $U_{j+1}=\set{i+1}\times V_u$.
  By choice of the matching, there is indeed an edge from $(i,\phi(u))$ to 
  $(i+1,\phi(u))$ in~$B_j$.
  Otherwise, $B_j$ is one of the graph copies, say $L_j=\set{i}\times 
  V_u$ and $L_{j+1}=\set{i}\times V_{v}$; note that by construction of $G'$, we 
  only do this if $\{u,v\}$ is an edge of $H$ since $v$ is the successor of $u$ in an 
  Eulerian tour of $H$.
  Since $\phi$ is a graph homomorphism respecting the coloring, we have that 
  $\phi(u) \phi(v)$ is an edge in $G[V_u\cup V_v]$.
  Since $B_j$ was a copy of this graph by construction, there is an edge from 
  $(i,\phi(u))$ to $(i,\phi(v))$.
  Overall, we get that $spt$ is an $s,t$-path in $G'$, and its length is 
  $\ell+2=2k+4=k'$; this is the canonical path corresponding to~$\phi$.

  In summary, every vertex-colorful $H$-subgraph $H'$ defines a unique canonical $s,t$-path.
  Conversely, let $p$ be a path whose internal vertex set is of the form 
  $V(p) = \set{(i_j,x_j)}$ where $x_j\in V_u$ for layer  $L_j=\set{i}\times V_u$, and that is consistent in the sense that it always picks the same vertex from each color class.
  (Formally, if $u_i=u_j$ then $x_i=x_j$.)
  Then $p$ describes a set of $|V(H)|$ many vertices and $k$ edges in $G$ that make up a colorful $H$-subgraph $H'$ of $G$.
  Since $H$ has no nontrivial automorphisms, this graph is uniquely defined.

  \medskip

  Let $\mathcal P$ be the set of all $s,t$-paths whose length~$\ell$ satisfies $k'\le\ell\le f(k')$.
  The central claim is that the number of paths in~$\mathcal P$ that are not canonical solutions is even.
  For this, we construct a fixed-point free involution~$\pi$ on non-canonical 
  solutions.
  
  First we concentrate on solutions that are \emph{jumpy}:
  A path $p\in\mathcal P$ is jumpy if
  there exist $i,j\in\set{1,\dots,k}$ with $i\ne j$ so that we added gadget 
  edges between $M_i$ and $M_j$, and
  \begin{enumerate}[(i)]
    \item
      $p$ uses the edge from $(i,x)\in L(M_i)$ to $(j,x)\in L(M_j)$ but not the 
      edge from $(j+1,x)\in R(M_j)$ to $(i+1,x)\in R(M_i)$,
      or
    \item
      $p$ uses the edge from $(i+1,x)\in R(M_i)$ to $(j+1,x)\in R(M_j)$ but not 
      the edge from $(j,x)\in L(M_j)$ and $(i,x)\in L(M_i)$.
  \end{enumerate}
  If $p$ is jumpy, we define $\pi(p)$ as follows: First we identify the 
  lexicographically first position where $p$ is jumpy.
  Then we exchange state (i) with state (ii) at that position.
  Note that (i) implies that $p$ uses the edge from $(j,x)$ to $(j+1,x)$ in 
  $M_j$ and (ii) implies that $p$ uses the edge from $(i,x)$ to $(i+1,x)$ in 
  $M_i$; we swap these edges too, when applying $\pi$.
  Now $\pi$ is a fixed-point free involution on jumpy paths, and note that 
  $\pi(p)$ has the same length as $p$.

  Next we consider paths that are \emph{bad}.
  A path $p\in \mathcal P$ is bad if it is not jumpy, and yet there exist 
  $i,j\in\set{1,\dots,k}$ with $i\ne j$ so that we added gadget edges between 
  $M_i$ and $M_j$, and
  \begin{enumerate}[(i)]
    \item
      $p$ uses the $M_i$-edge from $(i,x)$ to $(i+1,x)$ and not the $M_j$-edge 
      from $(j,x)$ to $(j+1,x)$, or
    \item
      $p$ uses the three edges from $(i,x)$ to $(j,x)$ to $(j+1,x)$ to 
      $(i+1,x)$.
  \end{enumerate}
  We define $\pi(p)$ for a bad path~$p$ by finding the first bad position, and 
  switching between (i) and (ii).
  Say, $p$ was in state (i) at that position, then $\pi(p)$ is in state (ii) at 
  that position, which makes the overall path exactly two edges longer.
  We claim that all non-jumpy paths pair up in this manner, without having to 
  consider arbitrarily long paths.
  Indeed, since $p$ is not jumpy, when we look at the vertices that~$p$ 
  traverses, we are essentially traversing the layers in a monotone order, except for potential short two-vertex detours in the case of gadgets in 
  a bad state~(ii).
  More precisely, when $p$ reaches a vertex in a layer~$L_j$, either the next 
  vertex is in $L_{j+1}$ or the third vertex after it is in $L_{j+1}$.
  This means that the paths moves to the right by one layer at least once ever 
  $4$ vertices, and thus the longest path that is not jumpy has length at most 
  $k+3(k+2)+2=k'$, accounting for $k$ graph edges, at most~$3$ matching/gadget edges 
  for each of the $k+2$ gadgets, and two edges at the source and sink.
\end{proof}

For completeness, we include two simple reductions:
First from the flexible-length to the fixed-length problem in directed graphs, then from the directed to undirected problem.
\begin{lemma}\label{lem:dipath-flexibility}
  Let $f:\N\to\N$ be a computable, strictly increasing function.
  There is a parsimonious polynomial-time fpt-reduction from \pp{\#Directed $f$-Flexible Paths} to \pp{\#Directed 
  Paths}.
\end{lemma}

\begin{lemma}\label{lem:dipath-to-path}
  There is a parsimonious polynomial-time fpt-reduction from \pp{\#Directed Paths} to \pp{\#Paths}.
\end{lemma}

With all these prerequisites collected, we can complete the proof.

\begin{proof}[Proof of \cref{thm: hardness-paths}]
  Let $\mathcal H$ be the class of all graphs~$H$ that are connected, $4$-regular, and whose automorphism group has size one.
  We use the probabilistic method to argue that the tree-width of graphs in~$\mathcal H$ is not bounded.
  With probability $1-o(1)$ as $h\to\infty$, random $4$-regular graphs with~$h$ vertices are connected~\cite[Theorem~2.10]{wormald1999models},
  they have no nontrivial automorphisms~\cite{DBLP:journals/rsa/KimSV02},
  and they are almost Ramanujan~\cite[Theorem~7.10]{hoory2006expander}, that is, their second-largest eigenvalue in absolute value satisfies $\lambda\le 2\sqrt{3}+o(1)<3.5$.
  By a union bound, $H$ has all three properties simultaneously with probability $1-o(1)$.
  By Cheeger's inequality~\cite[Theorem 4.11]{hoory2006expander}, we have 
  \[ \min_{S\subseteq V(H), |S|\le\frac12 h } \frac{\abs{E(S,\overline S)}}{\abs{S}}\ge \frac{1}{2}(4-\lambda)>0.1\,,\] that is, the edge expansion is bounded away from zero, which implies that the tree-width of~$H$ is at least linear in~$h$ (see, e.g.,~\cite[Exercise~7.34]{DBLP:books/sp/CyganFKLMPPS15}).
  Thus, $\mathcal H$ contains graphs of arbitrarily large tree-width.

  By \cref{cor: hardness-tw}, $\oplus\vcSub(\mathcal H)$ is $\pW$-hard under parsimonious fpt-reductions.
  If there is a parsimonious fpt-reduction from problem $\#A$ to problem $\#B$ then in particular the parity version $\oplus A$ reduces to $\oplus B$.
  Writing $\oplus A\leq \oplus B$ we can summarize the chain of reductions in \cref{lem:vcsubs-to-dipaths,lem:dipath-flexibility,lem:dipath-to-path} as
  \[
    \oplus\vcSub(\mathcal H)
    \le
    \oplus\pp{Directed $f$-Flexible Paths}\\
   \le
    \oplus\pp{Directed Paths}
    \le
    \oplus\pp{Paths}\,.
  \]
  This proves the $\pW$-hardness of $\oplus\pp{Paths}$.
  The containment follows from the standard fpt-reduction from $\#\pp{Paths}$ to $\#\pp{Clique}$, which is parsimonious.
  Overall, the claim follows.
\end{proof}


\section{Conclusion}

We conducted an initial investigation of modular subgraph counting, leading to the partial classification depicted in \cref{fig: overview}. To obtain a complete picture, the following conjecture needs to be addressed.

\begin{conjecture}\label{conjecture}
  For any computable pattern class $\mathcal H$:
  \begin{itemize}
    \item If $\mathcal H$ has unbounded matching-split number, then the problem $\oplus\Sub(\mathcal H)$ is $\pW$-complete.
    \item If $\mathcal H$ has unbounded vertex-cover number, then $\ModSub{\mathcal H}{q}$ for fixed $q \in \N$ is $\Mod_p \W$-complete for any odd divisor $p$ of $q$.
  \end{itemize}
\end{conjecture}

An appropriate transfer of the subgraph-homomorphism framework to modular counting is likely to help in settling this conjecture.

\bibliographystyle{plainurl}
\bibliography{main}

\appendix

\section{Omitted proofs from \cref{sec: preliminaries,sec: hardness-matchings}}

\begin{proof}[Proof of~\cref{lem:decision-to-parity}]
  The reduction repeats the following $O(2^{\abs{E(H)}})$ times:
  Obtain~$G'$ from~$G$ by deleting each edge independently with probability $1/2$, and query the oracle for~$\#\vcSub(H,G')\bmod q$. If all queries return~$0$, then return `no'. Otherwise, return `yes'.

  First note that the $H$-coloring $c$ for $G$ that is given as additional input is also an $H$-coloring of~$G'$, so the queries are valid.
  If~$G$ does not have colorful $H$-isomorphic subgraphs, then~$G'$ does not have such subgraphs either, and our reduction correctly returns `no'.
  For the other direction, suppose now that $G$ has at least one colorful $H$-isomorphic subgraph.
  We define a multilinear degree-$\abs{E(H)}$ polynomial over the variables $x_e$ with $e\in E(G)$:
  \begin{equation*}
    P(x)=\sum_{F} \prod_{e\in E(F)} x_e\,.
  \end{equation*}
  Here, the sum is over all colorful subgraphs~$F$ of~$G$ that are isomorphic to~$H$.
  The total number of colorful $H$-isomorphic subgraphs is the evaluation of this polynomial at the all-ones vector, and since at least one such colorful subgraph exists, $P(x)$ is not identical to~$0$.
  Moreover, choosing a random subgraph~$G'$ corresponds naturally to sampling all $x_e\in\set{0,1}$ independently and uniformly at random.
  In that case, $P(x)$ is equal to the number of colorful~$H$-isomorphic subgraphs of~$G'$.
  By~\cite[Lemma~2.1]{DBLP:conf/soda/WilliamsWWY15}, we then have with probability at least $2^{-\abs{E(H)}}$ over the random choice of~$x$ that $P(x)$ is not a multiple of~$q$, and so this is the probability over the random choice of~$G'$ that $\#\vcSub(H,G')\bmod q$ is not equal to~$0$.
  Repeating this $O(2^{\abs{E(H)}})$ times reduces the error probability to an arbitrarily small constant.
\end{proof}

\begin{proof}[Proof of~\cref{cor: hardness-tw}]
The problem $\vcSub(\classH)$ parameterized by $k=\abs{E(H)}$ is known to be $\W$-hard~\cite{DBLP:journals/tcs/DalmauJ04}.
  Moreover, if ETH is true, then the problem does not have an algorithm running 
  in time $n^{o(k/\log k)}$, where $n=\abs{V(G)}$, as shown by Marx~\cite{DBLP:journals/toc/Marx10}.

  The first claim of the lemma now follows, because the parameterized reduction from the $k$-clique problem to $\vcSub(\classH)$ in~\cite{DBLP:journals/toc/Marx10} is parsimonious.
  See, for example,~\cite[Theorem~5.6]{DBLP:phd/dnb/Curticapean15} for a self-contained proof.
  
  We prove the second claim: By~\cref{lem:decision-to-parity}, there is a randomized Turing fpt-reduction from $\vcSub(\classH)$ to $\#\vcSub(\classH)\bmod q$ that on input $(H,G)$ only makes queries $(H,G')$ with the same pattern~$H$. Thus, an $n^{o(k/\log k)}$-time algorithm for $\#\vcSub(\classH)\bmod q$ would imply an algorithm with asymptotically the same running time for the decision problem, which however is ruled out under ETH:
  Let $\classH$ be a graph family of unbounded tree-width.
  Then the problem $\vcSub(\classH)$ parameterized by $k=\abs{E(H)}$ is $\W$-hard.
  Moreover, if ETH is true, then the problem does not have an algorithm running 
  in time $n^{o(k/\log k)}$, where $n=\abs{V(G)}$.
\end{proof}

\begin{proof}[Proof of~\cref{claim: irrelevant-matchings}]
Call a colorful matching $M$ in $G'$ \emph{consistent} if every consistency gadget in $G'$ has either $0$ or $3$ edges in $M$.
We first prove that non-consistent matchings cancel modulo $q$.
For example, suppose that gadget $Q_u$ for some $u \in V_a$ for $a\in V(H)$ only matches $\set{u_1}$. In this case, since colorful matchings include exactly one edge from each of the colors $(\mathtt{CONS},a,i)$ for $i \in \set{1,2,3}$, one of the following occurs:
\begin{enumerate}
  \item The consistency gadget $Q_v$ for some other vertex $v\in V(G)$ only matches $\set{v_2, v_3}$. By~\eqref{eq: consistency-gadget-extensions}, the number of such matchings is $0$ modulo $q$.
  \item For some vertices $v,w \in V(G)$, the gadget $Q_v$ only matches $\set{v_2}$ and $Q_w$ only matches $\set{w_3}$. Again by~\eqref{eq: consistency-gadget-extensions}, the number of such matchings is $0$ modulo $q$.
\end{enumerate}
The same argument applies when $Q_u$ only matches $\set{u_2}$ or $\set{u_3}$.
It follows that non-consistent matchings cancel modulo $q$, 
and we can restrict our attention to the set of consistent matchings.
This set can be partitioned into sets $\mathcal{D}_X$ for $X \subseteq V(G)$ with $|X \cap V_a| = 1$ for $a\in V(H)$: 
The set $\mathcal{D}_X$ contains all matchings that match $\set{u_1,u_2,u_3}$ within $Q_u$ for all $u \in X$.
We show that $|\mathcal{D}_X| \equiv_q 0$ if $X$ does not correspond to the vertex set of an $H$-copy:
If there is some edge $\{a,b\} \in E(H)$ such that the unique vertices $u \in X \cap V_a$ and $v \in X \cap V_b$ are not connected in $G$, then the number of edges of color $(\mathtt{AND},ab)$ that extend matchings of consistency gadgets is $0$ modulo $q$.
It follows that $|\mathcal{D}_X| \equiv_q 0$ if $X$ is not the vertex set of an $H$-copy in $G$, and we conclude that matchings $M \notin \bigcup_{F} \mathcal{M}_F$ indeed cancel modulo $q$.
\end{proof}

\begin{proof}[Proof of \cref{thm: hardness-matchings}]
  By \cref{lem:reduction-to-colmatch}, it suffices to reduce the problem ${\#\pp{ColMatch}\bmod q}$ to counting $k$-matchings modulo $q$.
  Let $(G,c)$ be an input for the reduction, where $c\colon E(G)\to\mathcal C$ and $\abs{C}=k$.
  By inclusion-exclusion,
  \[
    \#\pp{ColMatch}(G,c) = \sum_{S\subseteq \mathcal C} \paren{-1}^{\abs{S}}\cdot\#\pp{Match}(G_S,k)\,,
  \]
  where $G_S$ is derived from $G$ by deleting all edges~$e$ with $c(e)\in S$ and $\#\pp{Match}(G_S,k)$ denotes the number of $k$-matchings in~$G_S$.
  This reduction runs in time $2^k\poly(n)$ and all queries use the same parameter~$k$.
  The claim follows.
\end{proof}

\section{Omitted material from \cref{sec: algorithm}}

  \begin{proof}[Proof of~\cref{lem:rigidize}]
    Step R1 can be performed by inspecting all subsets of $H$ of size at most $s$ in time $h^{O(s)}$; this part of the overall 
    algorithm for modular counting will not be the bottleneck, and we make no 
    attempt at improving the running time.
    Steps R2 and R3 can be carried out in time $\poly(h)$.
    This proves the claim for the running time.
  
    Since $R\subseteq T$ holds, it suffices to show $\abs{T}\le O(s^2)$ 
    for the first claim.
    Indeed, we have $\abs{S}\leq s$ and each vertex of $S\cap D$ has at 
    most $s+1$ neighbors~$w\in\overline S$.
    Moreover, each such~$w$ has at most one other neighbor in $\overline S$ since 
    $H-S$ has maximum degree one.
    Overall, R2 adds at most $\abs{S\cap D}\cdot 2(s+1)$ neighbors to~$T$, 
    so we can bound the size of the extended splitting set as
    \begin{equation*}
    \abs{T}
    \leq \abs{S} + \abs{S\cap D} \cdot 2(s+1)
    \leq O(s^2)
    \,.
    \end{equation*}
    This proves the first claim.
  
    \medskip
    For the second claim, we first note that $H-T$ is a matching, because~$S$ is a splitting set 
    and $S\subseteq T$.
    Let $C$ be a component of $H-R$.
    If ${V(C)\cap T=\emptyset}$, then $C$ is a component of $H-T$ and therefore has at most 
    two vertices.
    Otherwise, some vertices of $C$ must have been removed from~$R$ by Step~R3.
    The only vertices that R3 removes from~$R$ belong to~$T\cap D$, which has no 
    neighbors in~$\overline T$ because of Step~R2.
    Therefore, $V(C)$ is entirely contained in $T\cap D$.
    But then $C$ is a component of $H[T\cap D]$ that got removed from~$R$ by Step~R3, which implies that~$C$ has at most two vertices.
    This proves the second claim.
    
    \medskip
    For the third claim, let $f$ be an automorphism of~$H$.
    We want to show $f(R)=R$.
    We first observe that $\deg f(v) = \deg v$:
    Since $f$ is a homomorphism, if $u$ is a neighbor of $v$ then $f(u)$ is a neighbor of $f(v)$, and 
    since $f$ is injective, the neighbors $f(N(v))$ are all different.
 
    We proceed to show $f(v)\in R$ for each $v\in R$.
    There are two cases.
    First, consider $v \in R \setminus D$.
    We first show that $f(v)$ belongs to the splitting set $S$.
    Indeed, a vertex~${w\notin S}$ belongs to the matching $H\setminus S$ and therefore has at most one neighbor outside $S$. 
    In particular $\deg w\leq |S| + 1 = s + 1$.
    But ${\deg f(v) = \deg v > s + 1}$, because $v\notin D$.
    Thus, $f(v)\in S\subseteq T$.
    Since Step~R3 only removes small-degree vertices, the vertex $f(v)$ remains in the refined set~$R$.
  
    Now consider $v \in R \cap D$.
    Let~$C$ be the component of~$v$ in the graph $H[R\cap D]$.
    We have $\abs{V(C)}\geq 3$, since~$C$ would otherwise have been removed from~$R$ by R3.
    Since $C$ contains only vertices from~$D$, its image under~$f$ also satisfies 
    $f(V(C))\subseteq D$.
    Since there is no edge from $T\cap D$ to $\overline T$ because of Step~R2, the image $f(V(C))$ is either entirely contained in $T\cap D$ or in $\overline T$.
    Since $f(V(C))$ has three or more vertices and $T$ is a splitting set, the latter 
    cannot be the case.
    Finally, $C$ is too large to have been removed from $T$ in Step~R3, so $f(V(C))\subseteq R$ and thus $f(v)\in R$.
  
    We conclude $f(R)\subseteq R$, which proves the third claim.
  \end{proof}

  \subsection{Counting matchings over polynomial rings mod $2^t$}%
  \label{sec:Hafnian}
  
  In this section, we state an algorithm for computing the Hafnian over polynomial rings modulo~$2^t$, which we use in the proof of \cref{lem:color-demands}.
  Note that the related problem of computing permanents modulo prime powers has also been studied~\cite{BHL15};
  permanents over the quotient ring $\Z_4 [X] / ( X^m)$ were used to obtain a polynomial-time algorithm for finding disjoint paths of minimum total length~\cite{DBLP:journals/siamcomp/BjorklundH19}.
  
  Let $n$ be an even integer and let $A$ denote a symmetric $n\times n$ matrix with 
  entries $(a_{ij})$.
  The \emph{Hafnian} of $A$, denoted by $\haf A$, is defined by
  \[ \haf A = \sum_M \prod_{\{i,j\}\in M} a_{ij}\,,\]
  where the sum is over all partitions $M$ of $\{1,\ldots, n\}$ into sets of size $2$.
  If $A$ denotes a weighted adjacency matrix of a graph, then~$M$ contributes to $\haf A$ only if $M$ is a perfect matching.
  
  We consider this expression when $A$ is a matrix of $r$-variate polynomials with integer coefficients modulo $q$.
  In other words, the sum and product in the definition of the Hafnian are computed in the ring $\Z_q[X_1,\ldots, X_r]$.
  We use $\haf A\bmod q$ to denote this polynomial.
  Let us remark that we consider the ring to be an abstract polynomial ring, so in particular $X_i^j\ne X_i^k$ holds whenever $j\ne k$.
  Moreover, we represent polynomials explicitly by listing all non-zero coefficients.

  For our purposes, it will be sufficient to consider Hafnians of matrices whose entries are linear forms (polynomials of degree at most $1$) in the variables $X_1,\ldots, X_r$, but the proof of the following theorem could also be extended to a setting where the matrix entries are bounded-degree polynomials.
  \begin{theorem}\label{thm:haf}
    Let $t$ and $r$ be positive integers and $q=2^t$.
    Given an $n\times n$ symmetric matrix~$A$, whose diagonal is zero and whose entries are $r$-variate linear forms in~$\Z_q[X_1,\ldots, X_r]$, there is an algorithm that computes $\haf A\bmod q$ in time $n^{O(tr)}$.
  \end{theorem}
  \begin{proof}
    For $r=1$, this result was established by Hirai and Namba~\cite[Theorem~2.1]{HiraiNamba}, 
    following~\cite{DBLP:journals/siamcomp/Valiant08,DBLP:journals/siamcomp/BjorklundH19}.
    Their algorithm runs in time $n^{O(t)}$ as required.
    For $r\ge 2$, we reduce the $r$-variate case to the univariate case using Kronecker substitution.

    For an integer $D$, later chosen to be sufficiently large, consider the Kronecker map
    \[
      \kappa\colon
      \Z_q[X_1,\ldots,X_r]\rightarrow
      \Z_{q}[X]\,.
    \]
    defined for each monomial as
    \begin{equation}\label{eq:kronecker}
      X_1^{e_1}X_2^{e_2}\cdots X_r^{e_r}\mapsto  X^{D^0e_1 + D^1e_2+\cdots+ D^{r-1}e_r}\,,
    \end{equation}
    and linearly extended to the entire polynomial ring.
    For any monomial $X_1^{e_1}X_2^{e_2}\cdots X_r^{e_r}$ with $e_i<D$ for all $i$, the exponent in~\eqref{eq:kronecker} satisfies $D^0e_1 + D^1e_2+\cdots+ D^{r-1}e_r< D^r$.
    For any given polynomial $p\in\Z_q[X_1,\ldots,X_r]$ whose indeterminates each have individual degree less than~$D$, the image $\kappa(p)$ can be computed in time $D^{O(r)}$.
    Likewise, if $p\in\Z_{q}[X]$ has degree at most $D^r-1$, the inverse $\kappa^{-1}(p)$ is well-defined and can be computed in time $D^{O(r)}$.
  
    Let $A=(a_{ij})$ be the input, that is, an $n\times n$ symmetric matrix of linear forms $a_{ij}$ from $\Z_q[X_1,\ldots, X_r]$ with $a_{ii}=0$ for all $i\in\{1,\dots,n\}$.
    Our algorithm for $\haf A\bmod q$ uses Kronecker substitution via the identity
    \[\haf A \bmod q=\kappa^{-1}(\haf\kappa(A)\bmod q)\,.\]
    This identity holds if the inverse $\kappa^{-1}$ of Kronecker substitution is well-defined for the final Hafnian $\haf\kappa(A)\bmod q$, which means we have to choose $D$ larger than the degree of $\haf A\bmod q$. This degree is at most $n/2$, so $D=n/2+1$ suffices.
    That is, our algorithm first applies the Kronecker map to each entry of $A$ to obtain the $n\times n$ matrix $\kappa(A)=(\kappa(a_{ij}))$ of polynomials from $\Z_{q}[X]$.
    This step takes time $O(n^2 D^r) \le n^{O(tr)}$.
    It then uses Hirai and Namba~\cite[Theorem~2.1]{HiraiNamba} on the resulting matrix of univariate polynomials of degree $\le D^r$ to compute $\haf(\kappa(A))\bmod q$, which takes time $(nD^r)^{O(t)}\le n^{O(tr)}$ as required.
    Finally, it computes the preimage of $\haf(\kappa(A))$ under Kronecker substitution, which takes time $D^{O(r)}\le n^{O(tr)}$.
    Overall, this algorithm computes $\haf A\bmod q$ in the required running time.
  \end{proof}
  
  \subsection{Counting colored matchings with demands mod $2^t$}\label{sec:color-demands}
  In this section, we prove \cref{lem:color-demands}, our algorithm for counting colored matchings with demands modulo~$2^t$.
  In a first step, we reduce from matchings to perfect matchings, that is, we reduce to an instance that demands zero isolated vertices.
  To this end, we use ODD/EVEN gadgets for the perfect matching problem (see~\cite[Lemma~2.24]{DBLP:phd/dnb/Curticapean15} following~\cite{DBLP:journals/siamcomp/Valiant08}) and moving the color demands for isolated vertices to the external edges of the gadget.
  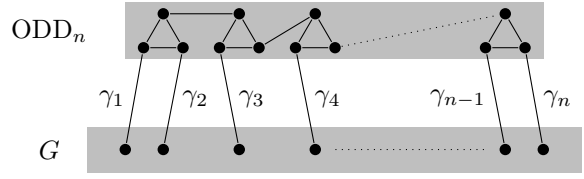
\begin{figure}[tp]
    \centering
    \begin{tikzpicture}
      \path[fill=lightgray] (-0.5,-0.8) rectangle (6,-0.2);
      \node at (-1,-0.5) {$G$};
  
      \path[fill=lightgray] (0,0.7) rectangle (5.5,1.45);
      \node at (-1,1.075) {ODD$_n$};
  
      \node[ci]                  (v1) at (0, -0.5)  {};
      \node[ci]                  (v2) at (0.5, -0.5)  {};
      \node[ci]                  (v3) at (1.5, -0.5)  {};
      \node[ci]                  (v4) at (2.5, -0.5)  {};
      \node[ci]                  (vnn) at (5,-0.5)  {};
      \node[ci]                  (vn) at (5.5,-0.5)  {};
  
      \node[ci]                  (t1) at ($(0.5,1)+(210:0.3)$)  {};
      \node[ci]                  (t2) at ($(0.5,1)+(330:0.3)$)  {};
      \node[ci]                  (t3) at ($(0.5,1)+(90:0.3)$)  {};
  
      \path[draw] (t1) -- (t2) -- (t3) -- (t1);
  
      \node[ci]                  (t4) at ($(1.5,1)+(90:0.3)$)  {};
      \node[ci]                  (t5) at ($(1.5,1)+(210:0.3)$)  {};
      \node[ci]                  (t6) at ($(1.5,1)+(330:0.3)$)  {};
      
      \path[draw] (t3) -- (t4) -- (t5) -- (t6) -- (t4);
  
      \node[ci]                  (t7) at ($(2.5,1)+(90:0.3)$)  {};
      \node[ci]                  (t8) at ($(2.5,1)+(210:0.3)$)  {};
      \node[ci]                  (t9) at ($(2.5,1)+(330:0.3)$)  {};
      
      \path[draw] (t6) -- (t7) -- (t8) -- (t9) -- (t7);
  
      \node[ci]                  (tn1) at ($(5,1)+(90:0.3)$)  {};
      \node[ci]                  (tn2) at ($(5,1)+(210:0.3)$)  {};
      \node[ci]                  (tn3) at ($(5,1)+(330:0.3)$)  {};
  
      \path[draw,dotted] (t9) -- (tn1);
      \path[draw] (tn1) -- (tn2) -- (tn3) -- (tn1);
  
      \path[draw,dotted,shorten >=4pt,shorten <=5pt] (v4) -- (vnn);
  
      \path[draw] (v1) edge node[left] {$\gamma_{1}$} (t1);
      \path[draw] (v2) edge node[right] {$\gamma_{2}$} (t2);
      \path[draw] (v3) edge node[right] {$\gamma_{3}$} (t5);
      \path[draw] (v4) edge node[right] {$\gamma_{4}$} (t8);
      \path[draw] (vnn) edge node[left] {$\gamma_{n-1}$} (tn2);
      \path[draw] (vn) edge node[right] {$\gamma_{n}$} (tn3);
  
    \end{tikzpicture}
    \caption{\label{fig:ODD-gadgets}%
    The gadget ODD$_n$ for \cref{lem:odd-even} is depicted at the top and the graph~$G$ with its~$n$ vertices at the bottom.
    Any odd-sized selection of edges $\gamma_i$ can be extended in a unique way to a perfect matching on ODD$_n$.}
  \end{figure}
  \begin{lemma}\label{lem:odd-even}%
    Given any instance $(G,c,D_I,D_E)$ with colors~$C$, we can compute in polynomial time $n + 1$ instances $(G',c',D_I',D_E^\ell)$ for $\ell\in\{0,\ldots, n\}$ with colors~$C'=C\dotcup\set{\ast}$, such that $D_I'(i)=0$ holds for all $i\in C'$ and $\abs{\mathcal M(G,c,D_I,D_E)}=\sum_{\ell\in\set{0,\dots,n}}\abs{\mathcal M(G',c',D_I',D_E^\ell)}$. 
  \end{lemma}
  \begin{proof}
    Without loss of generality, let $V(G)=\set{1,\dots,n}$.
    We sketch the proof in the case that~$n\ge 3$ is an odd integer.
    We construct~$G'$ as depicted in \cref{fig:ODD-gadgets}, by starting with~$G$ and adding the gadget ODD$_n$ from~\cite[Lemma~2.24]{DBLP:phd/dnb/Curticapean15} to it: In particular, we add a sequence of $n-2$ disjoint triangles connected to each other in a line; this leaves $n$ vertices of degree~$2$ in the gadget. We join the~$j$-th such vertex of degree~$2$ with the vertex $j\in V(G)$ using an edge~$\gamma_{j}$.
    This concludes the definition of~$G'$.
  
    The gadget ODD$_n$ has the property that every partial matching on its external edges can be extended in at most one way to a perfect matching on its internal edges.
    In particular, such an extension is possible if and only if the partial matching on the external edges has an odd number of edges.
  
    We define $c'(v)=\emptyset$ for all $v\in V(G')$, and we let $C'=C\dotcup\set{\ast}$, where $\ast$ is the only new color we add to~$C$. For all $e\in E(G')$, we define
    \begin{align*}
      c'(e)=
      \begin{cases}
        c(e);& \text{if $e\in E(G)$;}\\
        c(j);& \text{if $e=\gamma_j$ for some~$j\in V(G)$;}\\
        \set{\ast};& \text{if $e$ is an internal edge of ODD$_n$.}
      \end{cases}
    \end{align*}
    We define $D_I'(i)=0$ and $D_E^\ell(i)=D_I(i)+D_E(i)$ for all $i\in C$ and $\ell\in\set{0,\dots,n}$.
    Moreover, we set $D_E^\ell(\ast)=\ell$.
    We now have $\abs{\mathcal M(G,c,D_I,D_E)}=\sum_\ell\abs{\mathcal M(G',c',D_I',D_E^\ell)}$: Every colored matching~$(M,c_M)$ in $\mathcal M(G,c,D_I,D_E)$ with~$\ell$ isolated vertices can be extended in a unique way to a colored perfect matching from $\mathcal M(G',c',D_I',D_E^\ell)$, and this mapping is surjective.
  
    When $n$ is not an odd integer, the result follows in an analogous way by using an EVEN$_n$ gadget. (Note that no single gadget can ``absorb'' both even and odd numbers of unmatched vertices. Hence we need to distinguish between even and odd $n$ and choose gadgets accordingly.)
  \end{proof}
  
  We are now in position to prove \cref{lem:color-demands}, our algorithm for computing modulo~$2^t$ the number of colored matchings that satisfy a given demand.
  \begin{proof}[Proof of \cref{lem:color-demands}]
    By \cref{lem:odd-even}, we can assume without loss of generality that the instance~$(G,c,D_I,D_E)$ with colors~$C$ satisfies $D_I(i)=0$ for all~$i\in C$.
    Then $\mathcal M(G,c,D_I,D_E)$ contains only pairs $(M,c_M)$ where~$M$ is a perfect matching of~$G$.
    Suppose further that~$C=\set{1,\dots,r}$ holds for some~$r$, and define the $(n\times n)$-matrix~$A$ with
    \[
      A_{u,v}=
      \begin{cases}
        \sum_{i\in c(\set{u,v})} X_i;
        &\text{if $\{u,v\}\in E(G)$;}\\
        0;&\text{otherwise.}
      \end{cases}
    \]
    This matrix is symmetric, has zero diagonal, and the entries that are linear forms. Therefore, we can use \cref{thm:haf} to compute $\haf A$ as an element of $\Z_{2^t}[X_1,\dots,X_r]$, and we output the coefficient of the monomial $\prod_{i\in C} X_i^{D_E(i)}$.
    This algorithm runs in time $n^{O(tr)}=n^{O(t\abs{C})}$ as required.
  
    To see that the algorithm is correct, observe that every perfect matching~$M$ of~$G$ contributes the following term to the sum in~$\haf A$:
    \[
      \prod_{\{u,v\}\in M} \paren*{\sum_{i\in c(\set{u,v})} X_i}
      =
      \sum_{c'\colon E(M)\to C} \prod_{\{u, v\}\in M} X_{c'(\set{u,v})}\,.
    \]
    Thus, any pair $(M,c_M)$ satisfies the demands~$D_I,D_E$ if and only if the monomial corresponding to this pair is equal to~$\prod_{i\in C} X_i^{D_E(i)}$.
    Therefore, the coefficient of this monomial in~$\haf A$ is $\abs{\mathcal M(G,c,D_I,D_E)}\bmod 2^t$.
  \end{proof}

    We prove that ModCount satisfies the properties stated in \cref{thm: algorithm-split}.
    \begin{proof}[Proof of \cref{thm: algorithm-split}]
      We claim that ModCount$(H,G,t)$ counts the $H$-isomorphic subgraphs of $G$ modulo $2^t$ in time $n^{O(t4^s)}$.
    
      For the running time, note that C1 takes time $h^{O(s)}$ by \cref{lem:rigidize}.
      Moreover, C2 makes at most $n^s$ queries to \cref{lem:color-demands}, each of which takes time $n^{O(t\abs{C})}$.
      C3 takes time $n^{O(s)}$.
      Since $h\le n$ and $\abs{C}\le 2^s+2^s+2^{2s}\le O(4^s)$ hold, this leads to an overall running time of at most $n^{O(t4^s)}$ as claimed.
    
      For the correctness of the algorithm, it is sufficient to prove the identity:
      \begin{align}
        \#\Sub(H,G)
        &=
        \sum_{S\in\binom{V(G)}{s}}
        \sum_{\sigma\in E_S}
        \abs[\Big]{\mathcal M(G-S,c_{\sigma},D_I,D_E)}\,.
      \end{align}
      We prove this claim by constructing a bijection~$B$ between $H$-isomorphic subgraphs of~$G$ and 
      tuples $(S,\sigma,M,c_M)$ with $S\in\binom{V(G)}{s}$, $\sigma\in E_S$, and $(M,c_M)\in\mathcal 
      M(G-S,c_{\sigma},D_I,D_E)$.
      Let $F$ be a subgraph of~$G$ that is isomorphic to~$H$; we define 
      $B(F)=(S,\sigma,M,c_M)$ as follows.
    
      Let $\sigma'\in\Emb(H,F)$ be some embedding of $H$ in $F$, and let $S=\sigma'(R)$ be the image of the rigid splitting set~$R$ for~$H$.
      By assumption, because~$R$ is rigid, it is mapped to itself under any automorphism of~$H$. Since~$F$ and~$H$ are isomorphic, this implies that $S$ does not depend on the choice of the embedding~$\sigma'$.
    
      Next, we set $\sigma$ as the unique $\sigma\in E_S$ that is equivalent to $\sigma'\restrict[R]$ with respect to the group action of $\Aut(H)$ on $\Emb(H[R],G[S])$.
      Such a $\sigma$ exists and is unique by definition of~$E_S$.
    
      It remains to define the matching~$M$ and its vertex-coloring~$c_M$.
      We define~$M=F-S$.
      Since $S$ is a splitting set of~$F$, the graph~$M$ is indeed a matching.
      We define $N'_v=\sigma^{-1}(N_F(v) \cap S)$ for all $v\in V(M)$.
      For each isolated vertex $v\in I(M)$, we define the color~$c_M(v)$ via $c_M(v)=N'_v$.
      Moreover, for each isolated edge $\set{u,v}\in E(M)$, we define the color~$c_M(v)$ via $c_M(\set{u,v})=\set{N'_u, N'_v}$.
    
      We have now defined a mapping~$B$, and we need to argue that the image $B(F)$ of every $H$-isomorphic subgraph is a tuple~$(S,\sigma,M,c_M)$ of the appropriate type.
      Clearly, $S\in\binom{V(G)}{s}$ and $\sigma\in E_S$ hold by definition.
      It remains to prove that $(M,c_M)\in\mathcal M(G-S,c_{\sigma},D_I,D_E)$ holds.
    
      Indeed, $M$ is a matching in~$G-S$, that is, a subgraph of maximum degree~$1$.
      Moreover, $c_M$ sets colors that are permitted by $c_\sigma$:
      To see this, note that $N_v=\sigma^{-1}(N_G(v) \cap S)\supseteq \sigma^{-1}(N_F(v) \cap S)=N'_v$, and recall that $c_\sigma(v)$ contains all subsets~$N'$ of~$N_v$ and that $c_\sigma(\set{u,v})$ contains all sets $\set{N,N'}$ of subsets $N\subseteq N_u$ and $N'\subseteq N'_v$.
      Therefore, $c_M(o)\in c_\sigma(o)$ holds for all $o\in I(G)\cup E(G)$ and thus $c_M$ is a permissible coloring.
    
      We prove that $(M,c_M)$ satisfies the demands~$D_I,D_E$: Let~$N\subseteq R$ be arbitrary. We claim that $D_I(N)$ is equal to the number of isolated vertices~$v\in I(M)$ with $c_M(v)=N$.
      Since~$H$ and~$F$ are isomorphic and $M=F-S$ holds, the vertices of~$H-R$ and~$M$ are in a bijective correspondence that maintains the colors.
      More precisely, we can extend~$\sigma$ to an embedding~$\sigma''\in\Emb(H,F)$ such that all vertices $v\in V(H-R)$ and their respective images~$v'=\sigma''(v)\in V(F-S)$ satisfy $N_H(v)\cap R = \sigma^{-1}(N_F(v') \cap S)$.
      Even more, this~$\sigma''$ can be chosen so that it preserves the edge colors as well, that is, $\set{u,v}\in E(H-R)$ map to $\set{u',v'}\in E(F-S)$ under $\sigma''$ with the properties
      $N_H(u)\cap R = \sigma^{-1}(N_F(u') \cap S)$ and
      $N_H(v)\cap R = \sigma^{-1}(N_F(v') \cap S)$.
      It follows that $(M,c_M)$ satisfies the demands~$D_I,D_E$ and thus that $(M,c_M)\in\mathcal M(G-S,c_\sigma,D_I,D_E)$ holds and $B(F)=(S,\sigma,M,c_M)$ has the required type.
    
      Finally, it remains to argue that~$B$ is a bijection.
      For injectivity, let~$F$ and $F'$ be two $H$-subgraphs of~$G$, and let
      $B(F)=(S,\sigma,M,c_M)$ and $B(F')=(S',\sigma',M',c_{M'})$.
      Suppose that $B(F)=B(F')$ holds.
      We claim~$F=F'$.
      Indeed, since $S=S'$ and $\sigma=\sigma'\in E_S$ holds, we have 
      $F[S]=F'[S]$.
      Moreover, since $M=M'$ and $V(M)=V(F)\setminus S$, we have $F[V(M)]=F'[V(M)]$.
      It remains to argue that the edges between $V(M)$ and $S$ are identical in~$F$ 
      and $F'$.
      Let $v\in V(M)$.
      Since $c_M(v)=c_{M'}(v)$, we have $\sigma^{-1}(N_F(v)\cap S)=\sigma^{-1}(N_{F'}(v)\cap S)$. Since $\sigma$ is an isomorphism, this implies~$N_F(v)\cap S=N_{F'}(v)\cap S$. Overall, we get $F=F'$ as claimed, and $B$ is injective.
    
      For surjectivity of~$B$, let $(S,\sigma,M,c_M)$ be a tuple of the appropriate type, that is, $S\subseteq V(G)$ with $\abs{S}=s$, $\sigma\in E_S$, and $(M,c_M)\in\mathcal M(G-S, c_\sigma, D_I, D_E)$.
      We construct an $H$-isomorphic subgraph~$F$ of~$G$ with $B(F)=(S,\sigma,M,c_M)$.
      The vertex set of~$F$ is $S\cup V(M)$ and the edges on $V(M)$ are given by~$M$, so that $F=M-S$ holds.
      Since $\sigma\in E_S\subseteq \Emb(H[R], G[S])$, we can define edge set of~$F[S]$ as the image of the edge set of $H[R]$ under $\sigma$, that is, we have the edges $\set{\sigma(u),\sigma(v)}\in E(F[S])$ for all $\set{u,v}\in E(H[R])$.
      It remains to define the edge set between~$F[S]$ and $F-S$.
      For each $v\in V(F-S)$, we add the edges to all vertices in $\sigma(c_M(v))$.
      Now~$F$ is an $H$-isomorphic subgraph of~$G$ and $B(F)=(S,\sigma,M,c_M)$ holds, and we have proved the claim.
    \end{proof}

\section{Omitted proofs from \cref{sec: hardness-paths}}
    
    \begin{proof}[Proof of \cref{lem:dipath-flexibility}]
      Let $(G,s,t,k)$ be the input for the reduction.
      We construct the output $(G',s',t',k')$ as follows:
      Set $s'=s$ and $k'=f(k)+1$.
      We obtain $G'$ from $G$ by adding new vertices $v_1,\ldots, v_{k'-k}$ forming a directed path and all shortcuts from $t$.
      More precisely, we add the edges of the form $(t, v_j)$ and $(v_j,v_{j+1})$.
      \[
        \begin{tikzpicture}
          \node (s) at (1,0) [circle, fill, inner sep=2pt, label=left:{$s = s'$}] {};
          \node (t) at (3,0) [circle, fill, inner sep=2pt, label=below:$t$] {};
          \node (v_1) at (4,0) [circle, fill, inner sep=2pt, label=below:$v_1$] {};
          \node (v_2) at (5,0) [circle, fill, inner sep=2pt, label=below:$v_2$] {};
          \node  at (6,0)  {$\cdots$};
          \node (v_{k'-k-1}) at (7,0) [circle, fill, inner sep= 2pt, label=below:$v_{k'-k-1}$] {};
          \node (t_) at (8,0) [circle, fill, inner sep=2pt, label=right:{$t'=v_{k'-k}$}] {};
          \draw [-latex, dashed] (s) -- (t);
          \draw [-latex] (t) -- (v_1);
          \draw [-latex] (v_1) -- (v_2);
          \draw [-latex] (v_{k'-k-1}) -- (t_);
          \draw [-latex] (t) edge [bend left] (v_2);
          \draw [-latex] (t) edge [bend left] (v_{k'-k-1});
          \draw [-latex] (t) edge [bend left] (t_);
        \end{tikzpicture}
      \]
      We let $t'$ be the final vertex $v_{k'-k}$ on this path.
      We note that for all $j\in\{1,\ldots, k'-k\}$ there is a unique $t,t'$-path of length $j$.
    
      We now establish a bijection~$\pi$ between $s,t$-paths in~$G$ of some 
      length~$i\in\{k,\ldots, f(k)\}$ and $s',t'$-paths in $G'$ of length~$k'$:
      Let $p$ be an $s,t$-path of length~$i$, and set $j=k'-i$.
      Clearly $j\in\{1,\ldots, k'-k\}$ and so there is a unique $t,t'$-path~$q_j$ of 
      length~$j$,
      We define $\pi(p)= p q_j$.
      Then $\pi(p)$ is an $s',t'$-path of length exactly~$k'$, and the 
      function~$\pi$ is obviously injective.
      For the surjectivity, let $p'$ be any $s',t'$-path of length exactly~$k'$.
      Since $t$ separates $s$ from $t'$, $p'$ can be written at $p'=pq_j$ for some 
      $s,t$-path~$p$.
      Since $q_j$ is the only $t,t'$-path of length~$j$ and $j$ must be between 
      $1$ and $k'-k$, we have that $p$ has a length in $\{k,\ldots, k'-1\}$ so that 
      $\pi(p)=p'$.
      Thus, $\pi$ is bijective.
      This means that the input instance~$(G,s,t,k)$ for \pp{Directed $f$-Flexible Path} has the same 
      number of solutions as the output instance~$(G',s',t',k')$ for \pp{Directed Path}.
    \end{proof}
    
    \begin{proof}[Proof of \cref{lem:dipath-to-path}]
      The proof follows~\cite[Lemma~22]{FlumGrohe}, observing that their construction works also if their parameters are chosen as $p=1$ and $q=(k+1)$.
      This makes the argument somewhat simpler and allows a precise bound on the parameters of the resulting instance size.
    
      For a directed graph $G$, let $G'$ be the undirected graph obtained from~$G$ by the following steps:
      \begin{enumerate}
        \item
          Replace each vertex $v$ of $G$ by two adjacent vertices $v'$ and $v''$.
          Call such an edge type 1.
        \item 
          Replace each arc $\{u,v\}$ in $G$ by a $(k+1)$-path with endpoints $u''$ and $v'$. 
          Call such a path type 2.
      \end{enumerate}
      An $s,t$-path of length $k$ in $G$ corresponds to an $s',t'$-path of length $k(k+2)$ in $G'$.
      In fact, every $s',t'$-path in $G'$ that alternates between paths of both types of total length $\ell(k+2)$ corresponds to exactly one $\ell$-path in $G$.
      Conversely, a path in $G'$ that does not alternate between both types must use more type 2 paths, since type 1 paths are never adjacent.
      Thus, consider an $s',t'$-path in $G'$ that consists of $\ell_i$ paths of type $i\in\{1,2\}$, with $\ell_1<\ell_2$.
      Its length is \[
        \ell_1+\ell_2(k+1)  \begin{cases}
          \leq k-1 + k(k+1) < k(k+2)\,, & \text{if $\ell_2\leq k$}\,;\\
          \geq (k+1)(k+1)  = k^2 +2k+1 > k(k+2) \,,& \text{if $\ell_2>k$}\,.
        \end{cases}
      \]
      We conclude that the number of $s,t$-paths of length $k$ in $G$ equals the number of $s',t'$-path of length $k(k+2)$ in $G'$.
    \end{proof}
\end{document}